\documentclass[a4paper,onecolumn,11pt,accepted=2024-07-01]{quantumarticle}
\pdfoutput=1
\usepackage[utf8]{inputenc}
\usepackage[english]{babel}
\usepackage[T1]{fontenc}

\usepackage{tikz}
\usepackage{lipsum}

\usepackage[margin=1in]{geometry}

\usepackage{amssymb,amsmath,amsthm,amsfonts}
\usepackage{bm}
\usepackage{textgreek}
\usepackage{mathtools}
\usepackage{enumitem}
\usepackage[numbers,comma,sort&compress]{natbib}
\usepackage{authblk}
\usepackage{graphicx}
\usepackage[font=small]{caption}
\usepackage{subcaption} 
\usepackage{tablefootnote} 
\usepackage{float}
\usepackage[ruled,vlined,linesnumbered]{algorithm2e}
\usepackage{algorithmic}

\usepackage{physics}
\usepackage{footnote}
\usepackage{xcolor}
\usepackage{mathrsfs}
\usepackage{bbm}
\usepackage{makecell}
\usepackage{pbox}

\usepackage[colorlinks]{hyperref}

\newtheorem{theorem}{Theorem}
\newtheorem{definition}{Definition}
\newtheorem{lemma}{Lemma}
\newtheorem{proposition}{Proposition}

\newtheorem{corollary}{Corollary}

\newcommand{\thm}[1]{\hyperref[thm:#1]{Theorem~\ref*{thm:#1}}}
\newcommand{\cor}[1]{\hyperref[cor:#1]{Corollary~\ref*{cor:#1}}}
\newcommand{\defn}[1]{\hyperref[defn:#1]{Definition~\ref*{defn:#1}}}
\newcommand{\lem}[1]{\hyperref[lem:#1]{Lemma~\ref*{lem:#1}}}
\newcommand{\prop}[1]{\hyperref[prop:#1]{Proposition~\ref*{prop:#1}}}
\newcommand{\assum}[1]{\hyperref[assum:#1]{Assumption~\ref*{assum:#1}}}
\newcommand{\fig}[1]{\hyperref[fig:#1]{Figure~\ref*{fig:#1}}}
\newcommand{\tab}[1]{\hyperref[tab:#1]{Table~\ref*{tab:#1}}}
\newcommand{\algo}[1]{\hyperref[algo:#1]{Algorithm~\ref*{algo:#1}}}
\renewcommand{\sec}[1]{\hyperref[sec:#1]{Section~\ref*{sec:#1}}}
\newcommand{\append}[1]{\hyperref[append:#1]{Appendix~\ref*{append:#1}}}
\newcommand{\fac}[1]{\hyperref[fac:#1]{Fact~\ref*{fac:#1}}}
\newcommand{\lin}[1]{\hyperref[lin:#1]{Line~\ref*{lin:#1}}}

\def\>{\rangle}
\def\<{\langle}

\newcommand{\R}{\mathbb{R}}

\newcommand{\E}{\mathbb{E}}

\DeclareMathOperator{\poly}{poly}

\def\Tr{\operatorname{Tr}}\def\:{\hbox{\bf:}}

\newcommand{\tO}{\tilde{O}}

\newcommand{\cS}{\mathcal{S}}
\newcommand{\cV}{\mathcal{V}}
\newcommand{\cC}{\mathcal{C}}

\newcommand{\oH}{\overline{H}}
\newcommand{\oh}{\overline{h}}
\newcommand{\oW}{\overline{W}}

\newcommand{\blambda}{{\bm{\Lambda}}}
\newcommand{\bi}{{\bm{i}}}
\newcommand{\bj}{{\bm{j}}}

\newcommand{\plg}{\Pi_{>\Lambda}}
\newcommand{\pll}{\Pi_{\leq\Lambda}}
\newcommand{\plpg}{\Pi_{>\Lambda'}}
\newcommand{\plpl}{\Pi_{\leq\Lambda'}}

\newcommand{\pdl}{\Pi_{\leq\Delta}}

\newcommand{\pdpl}{\Pi_{\leq\Delta'}}

\newcommand{\pl}[1]{\Pi_{\leq\Lambda_{#1}}}

\let\oldnl\nl
\newcommand{\nonl}{\renewcommand{\nl}{\let\nl\oldnl}}

\usepackage{cleveref}
\renewcommand{\thefootnote}{\fnsymbol{footnote}}

\begin{document}

\title{Complexity of Digital Quantum Simulation in the Low-Energy Subspace: Applications and a Lower Bound}

\author{Weiyuan Gong$^{*}$}
\affiliation{John A. Paulson School of Engineering and Applied Sciences, Harvard University}
\affiliation{Institute for Interdisciplinary
Information Sciences, Tsinghua University}

\author{Shuo Zhou$^{*}$}
\affiliation{Center on Frontiers of Computing Studies, School of Computer Science, Peking University}
\affiliation{Yuanpei College, Peking University}

\author{Tongyang Li$^{\dagger}$}
\affiliation{Center on Frontiers of Computing Studies, School of Computer Science, Peking University}

\date{}
\maketitle

\begin{abstract}
Digital quantum simulation has broad applications in approximating unitary evolution of Hamiltonians. In practice, many simulation tasks for quantum systems focus on quantum states in the low-energy subspace instead of the entire Hilbert space. In this paper, we systematically investigate the complexity of digital quantum simulation based on product formulas in the low-energy subspace. We show that the simulation error depends on the effective low-energy norm of the Hamiltonian for a variety of digital quantum simulation algorithms and quantum systems, allowing improvements over the previous complexities for full unitary simulations even for imperfect state preparations {due to thermalization}. In particular, for simulating spin models in the low-energy subspace, we prove that randomized product formulas such as qDRIFT and random permutation require smaller Trotter numbers. Such improvement also persists in symmetry-protected digital quantum simulations. We prove a similar improvement in simulating the dynamics of power-law quantum interactions. We also provide a query lower bound for general digital quantum simulations in the low-energy subspace.\addtocounter{footnote}{1}\footnotetext[\value{footnote}]{Equal contribution.}\addtocounter{footnote}{1}\footnotetext[\value{footnote}]{Corresponding author. Email: \href{mailto:tongyangli@pku.edu.cn}{tongyangli@pku.edu.cn}}
\end{abstract}


\section{Introduction}\label{sec:Intro}

Quantum computers possess the potential to simulate the dynamics of given quantum systems more efficiently compared to their classical counterparts~\cite{georgescu2014quantum,feynman1982simulating,kim2023evidence}. A main approach for simulating Hamiltonian evolution is the digital quantum simulation, which maps the quantum systems into qubits and approximates the evolution using digitalized quantum gates. The most standard digitalization approaches are product formulas, including the Suzuki-Trotter formulas~\cite{suzuki1985decomposition,suzuki1991general,suzuki1990fractal,huyghebaert1990product,childs2019nearly,childs2004,berry2007efficient}, which decompose the target Hamiltonian into a sum of non-commuting terms and perform a product of exponentials of each individual term. Alternative principles have also led to advanced digital quantum simulation approaches such as linear combinations of unitaries~\cite{childs2012hamiltonian,zhuk2023trotter}, qubitization~\cite{low2019hamiltonian}, truncated Taylor series~\cite{berry2015simulating}, quantum signal processing~\cite{low2017optimal,martyn2021grand}, etc. Nevertheless, product formulas are still the most popular and promising approaches due to their relatively simple implementations, especially on near-term devices~\cite{preskill2018quantum,childs2018toward}.

Recent works provided refined analysis for standard product formulas~\cite{childs2021theory} and proposed improved digital quantum simulation algorithms including randomized compiling~\cite{campbell2019random,chen2021concentration,chen2024average}, random permutation~\cite{childs2019faster}, symmetry transformation~\cite{tran2021faster}, and multiproduct formula~\cite{low2019well,zhuk2023trotter,gong2023improved} based on product formulas. While most of these simulation algorithms are important and efficient for approximating Hamiltonian evolution, the resource requirement depends heavily on the norm of (the terms of) the Hamiltonian. However, a significant class of simulation tasks in quantum physics do not explore high-energy states, suggesting a potential to improve the resource demand under physically relevant assumptions on energy scales. As pointed out by Ref.~\cite{csahinouglu2021hamiltonian}, standard product formulas require smaller complexity to simulate the dynamics of $k$-local spin models when the states are restricted within a low-energy subspace. Hitherto, a general and complete theoretical framework for digital quantum simulations in the low-energy subspace still requires further exploration. {In addition, Ref.~\cite{PhysRevA.107.L040201} investigated eigenvector-dependent Trotter error which is closely related to simulating evolutions with quantum states chosen from a restricted subspace of the full Hilbert space, which is further explored by the concurrent work~\cite{burgarth2023strong}. Ref.~\cite{Yi2022Spectral} studied quantum simulation with low-dimensional initial state following spectral analysis.}

To this end, we systematically investigate the Hamiltonian simulation problem in the low-energy subspace. We analyze the Trotter number and gate complexity of randomized and symmetry-protected digital quantum simulation and prove the significant improvements compared to the best-known results for simulating states chosen from the full Hilbert space. We show that similar improvement can be obtained for quantum models such as the power-law model. In general, we provide a lower bound for Hamiltonian simulation in the low-energy subspace with logarithmic dependence on $\Delta$ and linear dependence on $t$. We further prove that the improvement above persists even when the quantum states are imperfectly prepared {due to thermalization}. Our results pave the way to more practical and efficient digital quantum simulations, which is an essential step toward the implementation of quantum simulation. {After we posted the original version of our work, Ref.~\cite{hejazi2024better} proposed improved bounds of low-energy simulation.}

\subsection{Our results}\label{sec:Results}
We begin with some basic concepts for digital quantum simulations. Given a target Hamiltonian $H=\sum_{l=1}^L H_l$ consisting of $L$ terms, we consider simulating an $n$-qubit quantum state $\rho$ undergoing a unitary evolution $V(t)=e^{-iHt}$ for some time $t$ using a unitary $U$. To achieve this goal, we split $t$ into $r$ (referred to as the Trotter number or the step complexity) identical Trotter steps such that the step size satisfies $\delta=t/r\ll 1$, and perform a product formula $W({\delta})=e^{-i\delta_q H_{l_q}}\cdots e^{-i\delta_1 H_{l_1}}$ in each time step, where $l_1,\ldots,l_q$ are chosen from $1,\ldots,L$. A standard instance of product formula is the Lie-Trotter formula
\begin{align*}
U(t)=S_1(\delta)^r=\left(\prod_{l=1}^L\exp(-iH_l\delta)\right)^r.
\end{align*}

Suzuki systematically extended the Lie-Trotter formula to a family of $p$-th order Suzuki-Trotter formulas defined recursively by
\begin{align*}
S_2(\delta)&=\prod_{l=1}^L\exp\Bigl(-iH_l \frac{\delta}{2}\Bigr)\prod_{l=L}^1\exp\Bigl(-iH_l\frac{\delta}{2}\Bigr),\\
S_p(\delta)&=S_{p-2}(u_p\delta)^2S_{p-2}((1-4u_p)\delta)S_{p-2}(u_p\delta)^2,
\end{align*}
where $u_p=(4-4^{1/(p-1)})^{-1}$~\cite{suzuki1985decomposition,suzuki1990fractal,suzuki1991general}. For large enough $r$, we have $V(t)\approx U(t)=S_{p}(\delta)^r$.

The main challenge for implementing product formulas is to decide a proper $r$ to guarantee that the simulation error is bounded below the error threshold $\epsilon$. In the general case, the simulation error for a unitary simulation $U$ is given by {$U\rho U^\dagger-V\rho V^\dagger$}. For technical simplicity, we consider the spectral norm $\norm{\cdot}$ as the error metric in this paper. As $\norm{{U\rho U^\dagger-V\rho V^\dagger}}\leq{2}\norm{U-V}$, we consider the spectral distance between the unitaries $U$ and $V$ in the following analysis. {We remark that the trace norm is also usually employed to measure the distance in density matrices. To upper bound the distance between the output states of two unitary evolutions $U\rho U^\dagger$ and $V\rho V^\dagger$, the diamond norm
\begin{align*}
\norm{\mathcal{U}-\mathcal{V}}_\diamond\coloneqq\max_{\ket{\psi}}\norm{U\ket{\psi}\bra{\psi}U^\dagger-V\ket{\psi}\bra{\psi}V^\dagger}_1,  
\end{align*}
which is defined as the maximal trace distance between the output states of two unitary evolutions $\mathcal{U}(\cdot)=U\cdot U^\dagger$ and $\mathcal{V}(\cdot)=V\cdot V^\dagger$ given any identical pure state as input. It is further proved that the diamond norm between two unitary channels can be bounded by (see e.g. Lemma 3.4,~\cite{chen2021concentration})
\begin{align*}
\norm{\mathcal{U}-\mathcal{V}}_\diamond\leq 2\norm{U-V}.
\end{align*}
Therefore, it is enough to consider the spectral norm $\norm{U-V}$ in the following.
}

Given a Trotter number $r$, the Lie-Trotter formula provides a
second-order simulation approximation error $O(\lambda_H^2 t^2/r)$, where $\lambda_H=\sum_{l=1}^L\norm{H_l}\leq L\lambda_h$ and $\lambda_h=\max_{l}\norm{H_l}$~\cite{childs2021theory}. Therefore, choosing the Trotter number $r_1=O(t^2\lambda_H^2/\epsilon)$ is sufficient to ensure the error below $\epsilon$. Similarly, choosing the Trotter number $r_p=O((\lambda_H t)^{1+1/p}/\epsilon^{1/p})$ is sufficient to ensure the error below $\epsilon$ for $p$-th order Lie-Suzuki-Trotter formulas.
\renewcommand{\thefootnote}{\arabic{footnote}}
\setcounter{footnote}{0}

It is sometimes possible to improve the complexity of product formulas with further information or assumptions in the form of Hamiltonians and states. In this paper, we consider digital quantum simulation in the low-energy subspace. Without loss of generality, we assume the Hamiltonian $H$ is composed of $L$ positive semi-definite terms $H_l\geq 0$.\footnote{Otherwise, we can consider simple shifts $H_l\to H_l+a_lI$ for all $H_l$'s for some positive $a_l$ such that $H_l+a_lI\geq 0$, which only adds an additional global phase to the evolution $V$.} For the simulation task, we only care about the low-energy quantum states that are below spectrum $\Delta\ll \norm{H}$. We denote the projector to this subspace as $\pdl$. The formal definition for the low-energy simulation is given as follows:
\begin{definition}\label{def:Prob}
Given a Hamiltonian $H$ with the above assumptions and a quantum state $\rho$ with energy lower than $\Delta\leq\norm{H}$, our goal is to find a simulation channel $\mathcal{U}(\cdot)$ such that the simulation error is below some threshold $\epsilon$, i.e., $\norm{\mathcal{U}(\rho)-{V\rho V^\dagger}}\leq\epsilon$, where $V=e^{-iHt}$ is the target evolution of $H$ with simulation time $t$.
\end{definition}
In the special case when $\mathcal{U}(\cdot)$ in \Cref{def:Prob} is a unitary channel for some unitary $U$, we only need to consider the norm $\norm{(U-V)\pdl}$ when estimating the upper bound for the simulation error between the simulation unitary $U$ and the target evolution $V$. A specific well-studied model is the spin Hamiltonian $H$ which has $k$-local interactions and strengths for each individual interaction term bounded below by $J$. Assuming each $H_l$ contains at most $M$ interaction terms and the number of interaction terms acting on any spin bounded by $d>0$, Ref.~\cite{csahinouglu2021hamiltonian} showed that choosing a Trotter number $r'=O((t\Delta)^{1+1/p}/\epsilon^{1/p}+(t\sqrt{n})^{1+1/(2p+1)}/\epsilon^{1/(2p+1)})$ for a $p$-th order Suzuki-Trotter formula suffices to achieve $\epsilon$ accuracy, where $n$ is the number of qubits. Compared to the simulation complexity for Suzuki-Trotter formulas for the general case simulation, we can obtain an improvement on $n$ dependence for $p=1$ as $\lambda_H=O(n)$ for $k$-local Hamiltonians. Ref.~\cite{csahinouglu2021hamiltonian} also showed that other improvements are possible for $p\geq 2$ when choosing the parameters $d$, $M$, and $J$ of the $H$ within some regimes. This result indicates that standard product formulas applied to low-energy states have smaller errors than the worst-case error for $k$-local spin Hamiltonians. 

\begin{table}[ht]
\centering
\resizebox{1.0\columnwidth}{!}{
\begin{tabular}{llll}
\hline
Model & Approach & Unitary simulation & Low-energy simulation \\ 
\hline
\hline
$k$-local & PF & $O\left(\frac{t^{1+1/p}}{\epsilon^{1/p}}n^{1/p}\right)$~\cite{childs2021theory} & \pbox{20cm}{$\tO\left(\frac{(t\Delta)^{1+1/p}}{\epsilon^{1/p}}\right)+$ \\ $O\left(\frac{t^{1+1/(2p+1)}}{\epsilon^{1/(2p+1)}}n^{\frac12+\frac{1}{4p+2}}\right)$~\cite{csahinouglu2021hamiltonian}}\\
\hline
$k$-local & qDRIFT & $O\left(\frac{t^2}{\epsilon}n^2\right)$~\cite{campbell2019random} & \bm{$\tO\left(\frac{t^2}{\epsilon}\Delta^2\right)+O\left(\frac{t^{4/3}}{\epsilon^{1/3}}n^{2/3}\right)$}\\
\hline
$k$-local & rand.~perm. & $O\left(\frac{L(tMJ)^{1+1/p}}{\epsilon^{1/p}}\right)$~\cite{childs2019faster}\tablefootnote{This paper considered bounds in the diamond norm instead of the spectral norm. However, we can derive the complexity for the spectral norm case directly through their results.} & \pbox{20cm}{\bm{$\tO\left(\frac{L(t\Delta)^{1+1/p}}{\epsilon^{1/p}}\right)+$}\\\bm{$O\left(\frac{t^{1+1/(2p+1)}}{\epsilon^{1/(2p+1)}}(LnJ)^{\frac12+\frac{1}{4p+2}}\right)$}}\\
\hline
$k$-local & sym.~prot. & $O\left(\left(\frac{nt^2}{\epsilon}\right)^{\frac{1}{1+\theta}}+\frac{nt^{3/2}}{\epsilon^{1/2}}\right)$~\cite{tran2021faster} & \pbox{20cm}{\bm{$\tO\left(\left(\frac{\Delta t^2}{\epsilon}\right)^{\frac{1}{1+\theta}}+\frac{\Delta t^{3/2}}{\epsilon^{1/2}}\right)$}+\\\bm{$O\left(\left(\frac{nt^3}{\epsilon}\right)^{\frac{1}{2+\theta}}+\frac{t^{5/4}n^{1/2}}{\epsilon^{1/4}}\right)$}}\\
\hline
power-law & PF & $O\left(\frac{(gt)^{1+1/p}}{\epsilon^{1/p}}n^{1/p}\right)$~\cite{childs2021theory} &  \pbox{20cm}{\bm{$\tO\left(\frac{(gt)^{1+1/p}}{\epsilon^{1/p}}\Delta^{1+1/p}\right)+$}\\\bm{$O\left(\frac{t^{1+1/(2p+1)}}{\epsilon^{1/(2p+1)}}(gn)^{\frac12+\frac{1}{4p+2}}\right)$}}\\
\hline
\end{tabular}
}
\caption{A summary of low-energy simulation complexity (Trotter number) for product formulas (PF), qDRIFT, random permutation (rand.~perm.), and symmetry protection (sym.~prot.) approaches, with comparisons to the state-of-the-art results. The models include $k$-local Hamiltonians and power-law models. 
The results in this work are marked in the \textbf{bold} font. Here, $\tO(\cdot)$ omits the polylogarithmic dependence on the parameters. For $k$-local Hamiltonians, we assume that $d$, $L$, and $J$ to be $O(1)$ except for the random permutation approach. For the symmetry protection approach, $0\leq\theta\leq1$ is a constant depending on the structure of the Hamiltonian, {see \eqref{eq:transformation}}. In the last line, $\alpha$ is the decay factor for the power-law interactions. Here, $g$ is the interaction strength and $g=O(1)$, $O(\log n)$, and $O(n^{1-\alpha/D})$ for $\alpha>D$, $\alpha=D$, and $\alpha<D$, where $D$ is the dimension of the system. For random approaches, we only list the results considering the simulation errors in expectation. We derive the Trotter number required for ensuring the random fluctuation in the following sections. In addition, we provide corresponding gate complexities for the results in the above table. {We can observe that  for all the approaches, low-energy simulation improves power dependence on $n,t,\epsilon$ in a trade-off or decoupling way and $n $ or spectral norm scaling is partially substituted by $\Delta$.}} 
\label{tab:main}
\end{table}

In this work, we conduct a systematic study of quantum simulation under the low-energy subspace assumption by proving the robustness of low-energy simulation under imperfect state preparation {due to thermalization}, proving the lower bound for low-energy simulations, analyzing the complexity required for implementing various digital quantum simulation approaches~\cite{campbell2019random,childs2019faster,tran2021faster}, and simulating different models such as $k$-local Hamiltonians and power-law models. We summarize our results in the following \Cref{tab:main} with comparisons to the previous results. We also carry out extensive numerical experiments to benchmark the low-energy simulation using various algorithms in physically relevant models.

\subsubsection{Robustness against imperfect state preparations {due to thermalization}}
Concerning experimental implementations of quantum simulations, it is hard to guarantee that the support for the input states is restricted exactly in the subspace below the threshold $\Delta$. In particular, the state preparation procedures on these near-term devices are imperfect, which produces a Gaussian tail on the higher energy states in the energy spectrum. Quantitatively, the probability distribution for eigenstates of energy $E$ higher than $\Delta$ satisfies a Gaussian distribution $\propto e^{-(E-\Delta)^2/\sigma^2}$, where $\sigma\ll\Delta$ is the variance{~\cite{banuls2020entanglement,ge2019faster,lu2021algorithms}}. In this work, we prove that such imperfect state preparation causes an $O((\sigma/\Delta)^2)$ multiplicative extra error in \Cref{sec:ImperfectPrep}, 
which is ignorable compared to the original low-energy simulation error. This result indicates that the low-energy improvement persists even under imperfect state preparations.

\subsubsection{Applications}
\paragraph{Randomized quantum compiling (qDRIFT).}
The first approach we analyze is the qDRIFT algorithm proposed by Ref.~\cite{campbell2019random}. We recall that the first-order Lie-Trotter algorithm deterministically cycles through every term $H_l$ in the Hamiltonian $H$ in each step. The qDRIFT algorithm~\cite{campbell2019random}, however, approximates the target evolution using a channel constructed by averaging products $U=U_rU_{r-1}\cdots U_1$. Here, each unitary $U_i$ is a short-time evolution on a single term $H_l$ sampled according to the probability distribution $\{p_l=\norm{H_l}/\lambda_H\}_l$. It is proved that only $O(\lambda_H^2t^2/\epsilon)$ steps suffice to ensure the error within $\epsilon$ in expectation, which is explicitly independent of $L$. 

In this work, we consider the performance of qDRIFT in the low-energy subspace. We denote $\oH=\pdpl H\pdpl$ for some carefully chosen $\Delta'\geq\Delta$ to be fixed later. We then introduce the parameter $\lambda_{\oH}=\sum_{l=1}^L\norm{\oH_l}$ where $\oH_l=\pdpl H_l\pdpl$. We now consider the low-energy version of qDRIFT where we evolve the system under $H_l\lambda_{\oH}/\norm{\oH_l}$ in each time step with probability $p_l=\norm{\oH_l}/\lambda_{\oH}$. Assuming the Hamiltonian $H$ is not ill-divided in the sense that $\lambda_{\oH}/\norm{\oH_l}=O(L)$ for all $l=1,\ldots,L$, we obtain the following complexity upper bound.
\begin{theorem}\label{thm:qDRIFTLow}
Let $H=\sum_{l=1}^LH_l$ be an $n$-qubit $k$-local Hamiltonian with parameters $M$, $J$, and $d$ defined as above ($LMJ=O(n)$). By choosing the Trotter number
\begin{align}\label{eq:qDRIFTLowExp}
r_{{\mathrm{exp}}}=\tO\left(\frac{L^2(\Delta+dkJ)^2t^2}{\epsilon}\right)+O\left(\frac{(LJ)^{4/3}M^{2/3}t^{4/3}}{\epsilon^{1/3}}\right),
\end{align}
we can ensure that the expected simulation error in the low-energy subspace below $\Delta$ for the qDRIFT algorithm is bounded by $\epsilon$, i.e., $\norm{(V-\E[U])\pdl}\leq\epsilon$. Moreover, if we choose the Trotter number
\begin{align}\label{eq:qDRIFTLowProb}
r_{{\mathrm{prob}}}=\tO\left(\frac{L^2(\Delta+dkJ)^2 t^2}{\epsilon^2}\left(n+\log(\frac{1}{\chi})\right)\right)+O\left(\frac{(LJ)^{4/3}M^{2/3}t^{4/3}}{\epsilon^{2/3}}\left(n+\log(\frac{1}{\chi})\right)^{1/3}\right),
\end{align}
then with probability at least $1-\chi$, we can guarantee that the simulation error for a single product $U$ is bounded by $\epsilon$.
\end{theorem}

We leave the full proof for \Cref{thm:qDRIFTLow} in \Cref{sec:Rand}. \Cref{thm:qDRIFTLow} provides the step complexity (Trotter number) required for simulating low-energy states of $k$-local spin Hamiltonians both in the expectation case and in the case when one needs to ensure the performance of every single simulation unitary $U$ with high confidence. Compared to the step complexity for implementing the qDRIFT algorithm in the full Hilbert space, the above theorem indicates that we can expect an improvement concerning the dependence on system size $n$ as long as the number of terms in each $H_l$ $M=\Omega(1)$, which can be satisfied by most of the spin models. Although the step complexity~\eqref{eq:qDRIFTLowExp} in expectation and the step complexity~\eqref{eq:qDRIFTLowProb} to control random fluctuations are no better than that for the low-energy simulation using standard product formulas, we remark that this deficiency originates from the disadvantage of qDRIFT in simulating $k$-local spins even in the full Hilbert space~\cite{chen2021concentration}.

\paragraph{Random permutation.}
Except for the qDRIFT algorithm, we consider the random permutation approach in simulating $k$-local Hamiltonians~\cite{childs2019faster}. While the qDRIFT approach implements random evolution on each term of the Hamiltonian, the random permutation approach randomly permutes the sequence of the terms in higher-order Suzuki-Trotter formulas. In particular, for any permutation in the permutation group $\sigma\in\cS_{L}$, we denote
\begin{align*}
S_2^\sigma(\delta)&=\prod_{l=1}^L\exp\Bigl(-iH_{\sigma(l)} \frac{\delta}{2}\Bigr)\prod_{l=L}^1\exp\Bigl(-iH_{\sigma(l)}\frac{\delta}{2}\Bigr),\\
S_p^\sigma(\delta)&=S_{p-2}^\sigma(u_p\delta)^2S_{p-2}^\sigma((1-4u_p)\delta)S_{p-2}^\sigma(u_p\delta)^2,
\end{align*}
where $u_p=(4-4^{1/(p-1)})^{-1}$. In each time step $\delta$, we consider averaging over all possible $\sigma\in \cS_{L}$. For the simulation in low-energy subspace, we estimate the simulation error using
\begin{align}\label{eq:RandPermErr}
\norm{\left[V(t)-\left(\frac{1}{L!}\sum_{\sigma\in \cS_L}S_p^\sigma(\delta)\right)^r\right]\pdl}.
\end{align}

In practice, it is difficult to implement all the permutations $\sigma$ and average the results as the number of all permutations $L!$ increases exponentially with $L$. We consider a random sampling procedure similar to the qDRIFT algorithm. We sample by implementing a random evolution $U=U_rU_{r-1}\cdots U_1$ with each $U_i=S_p^{\sigma_i}(\delta)$ a randomly chosen permuted Suzuki-Trotter formula and $\sigma_i\sim\cS_L$ i.i.d.~random permutations. In expectation, the evolution in each step is $\frac{1}{L!}\sum_{\sigma\in \cS_L}S_p^\sigma(\delta)$, which is the same as the standard random permutation approach. Concerning implementing this approach in the low-energy subspace, we obtain the following complexity result.
\begin{theorem}\label{thm:RandPermLow}
Let $H=\sum_{l=1}^LH_l$ be an $n$-qubit $k$-local Hamiltonian with parameters $M$, $J$, and $d$ ($LMJ=O(n)$). By choosing the Trotter number
\begin{align}\label{eq:RandPermLowExp}
r_{{\mathrm{exp}}}=\tO\left(\frac{Lt^{1+1/p}(\Delta+dkJ)^{1+1/p}}{\epsilon^{1/p}}\right)+O\left(\frac{Lt^{1+1/(2p+1)}(dkMJ^2)^{\frac12+\frac{1}{4p+2}}}{\epsilon^{1/(2p+1)}}\right),
\end{align}
we can ensure that the expected simulation error in the low-energy subspace below $\Delta$ for the random permutation approach is bounded by $\epsilon$. For the random sampling implementation of the random permutation algorithm, if we choose the Trotter number
\begin{align}\label{eq:RandPermLowProb}
r_{{\mathrm{prob}}}=\tO\left(\frac{(Lt(\Delta+dkJ))^{\frac{2p+2}{2p+1}}\left(n+\log\left(\frac{1}{\chi}\right)\right)^{\frac{1}{2p+1}})}{\epsilon^{\frac{2}{2p+1}}}\right)+O\left(\frac{(L^2t^2dkMJ^2)^{\frac{2p+2}{4p+3}}\left(n+\log\left(\frac{1}{\chi}\right)\right)^{\frac{1}{4p+3}})}{\epsilon^{\frac{2}{4p+3}}}\right),
\end{align}
then with probability at least $1-\chi$, we can ensure that the simulation error for a single random sequence is bounded by $\epsilon$.
\end{theorem}

We leave the technical details for \Cref{thm:RandPermLow} in \Cref{sec:Rand}. \Cref{thm:RandPermLow} shows the improvement over the step complexity of random permutation in the full Hilbert space~\cite{childs2019faster} concerning the dependence on systems size $n$. Compared to the step complexity in \eqref{eq:RandPermLowExp} with that for standard product formulas in the low-energy subspace~\cite{csahinouglu2021hamiltonian}, we can obtain an $L^{1/2p}$ reduction in the step complexity. This originates from the reduction from implementing the random permutation in the Hilbert space~\cite{childs2019faster}.

Based on the random permutation approach, Ref.~\cite{cho2022doubling} proposed an approach to double the order of approximation via the randomized product formulas. In particular, given $r\geq(5^{p/2-1}+\frac12)L\lambda_h t$ and $p$-th order contributing product formula, this algorithm provides a $(2p+1)$-th order simulation for the ideal evolution. For this approach, we also prove the following corollary on the Trotter numbers required for low-energy simulations.
\begin{corollary}\label{coro:DoubLow}
Let $H=\sum_{l=1}^LH_l$ be an $n$-qubit $k$-local Hamiltonian with parameters $M$, $J$, and $d$ ($LMJ=O(n)$). By choosing the Trotter number
\begin{align}\label{eq:DoubLowExp}
r_{{\mathrm{exp}}}=\tO\left(\frac{(Lt)^{1+1/(2p+1)}(\Delta+dkJ)^{1+1/(2p+1)}}{\epsilon^{1/(2p+1)}}\right)+O\left(\frac{(Lt)^{1+1/(4p+3)}(dkMJ^2)^{\frac{2p+2}{4p+3}}}{\epsilon^{1/(4p+3)}}\right),
\end{align}
we can ensure that the expected simulation error in the low-energy subspace below $\Delta$ for the doubling order approach is bounded by $\epsilon$. And the Trotter complexity $r_{\text{prob}}$ to ensure the algorithm converges remains the same as that in \Cref{thm:RandPermLow}.
\end{corollary}

\paragraph{Symmetry protection.}
Other than the randomized digital quantum simulation approaches, another approach to improve the performance of digital quantum simulations is to implement a symmetry transformation in each time step~\cite{tran2021faster}. 

Given a Hamiltonian $H$, we assume that it is invariant under a group of unitary transformations denoted by $\cC$. For each unitary $C$ chosen from $\cC$, we explicitly have $[H,C]\equiv 0$. The group $\cC$ represents a symmetry group of the system. According to Ref.~\cite{tran2021faster}, we ``rotate" each implementation of the Lie-Trotter formula $S_1(\delta)$ via a symmetry transformation $C_\mu\in\cC$ in each time step. The simulation for $V=e^{-iHt}$ reads $V(t)\approx\prod_{\mu=1}^rC_\mu^\dagger S_1(\delta)C_\mu$. Assuming the simulation error coherent, the digital quantum simulation of an evolution $V(t)$ may end up with the time evolution of a different Hamiltonian $H_{\text{eff}}$. We denote $\kappa=H_{\text{eff}}-H$ and rewrite the simulation under the assumption $\norm{\kappa}$ is small
\begin{align*}
V(t)\approx\prod_{\mu=1}^rC_\mu^\dagger S_1(\delta)C_\mu=\prod_{\mu=1}^re^{-iC_\mu^\dagger H_{\text{eff}}C_\mu \delta}=\prod_{\mu=1}^re^{-i(H+C_\mu^\dagger \kappa C_k)\delta}\approx e^{-i\left(H+\frac{1}{r}\sum_{\mu=1}^rC_\mu^\dagger \kappa C_\mu\right)t}.
\end{align*}

In the last line, we use the Baker-Campbell-Hausdorff (BCH) formula. The simulation error for this symmetry protection approach can thus be represented as
\begin{align*}
\norm{V(t)-\prod_{\mu=1}^LC_{\mu}^\dagger S_1(\delta)C_\mu}.
\end{align*}

For the symmetry protection approach, the Trotter number required for the $k$-local spin model is provided by Ref.~\cite{tran2021faster}
\begin{align*}
r=\max\left\{O\left(\left(\frac{nt^2}{\epsilon}\right)^{\frac{1}{1+\theta}}\right),O\left(\frac{nt^{3/2}}{\sqrt{\epsilon}}\right)\right\},
\end{align*}
where $\theta\in[0,1]$ is a constant that depends only on the structure of the Hamiltonian $H$ and the properties of the symmetry group $\cC$. {An intuitive definition of $\theta$ is given by considering the simulation error $v_0$ of the Lie-Trotter formula. The scaling of the ratio between $\norm{\mathbb{E}_{C_\mu\in\mathcal{C}}[C_\mu^\dagger v_0 C_\mu]}$ and $\norm{v_0}$ is $n^{\theta}$. Quantitatively, we have
\begin{align}\label{eq:theta}
\norm{\mathbb{E}_{C_\mu\in\mathcal{C}}[C_\mu^\dagger v_0 C_\mu]}=\norm{v_0}/n^{\theta}.
\end{align}}
There are several examples provided with a specific value of $\theta$. For example, when we draw symmetry transformations randomly, the behavior of the error would be analogous to a random walk, which results in $\theta=0.5$ under specific settings. The rigorous proof for this intuition for the localized Heisenberg model was provided in Ref.~\cite{tran2021faster}. Following a similar derivation in Ref.~\cite{gong2023improved}, one can explicitly obtain a construction where $\theta$ achieves the maximal value of $1$.

Now, we consider the performance of this symmetry protection approach in the low-energy subspace, which is also an open question in Ref.~\cite{tran2021faster}. We estimate the simulation error by computing the following error term
\begin{align*}
\norm{\left(V(t)-\prod_{\mu=1}^rC_\mu^\dagger S_1(\delta)C_\mu\right)\pdl}.
\end{align*}

We prove the following theorem concerning the low-energy simulations of $k$-local Hamiltonians using the symmetry protection approach.
\begin{theorem}\label{thm:SymProtLow}
Let $H=\sum_{l=1}^LH_l$ be an $n$-qubit $k$-local Hamiltonian with parameters $M$, $J$, and $d$ ($LMJ=O(n)$). By choosing the Trotter number
\begin{align}\label{eq:SymProtLowExp}
r=\tO\left(\left(\frac{L(\Delta+JdkL)t^2}{\epsilon}\right)^{\frac{1}{1+\theta}}+\left(\frac{L^2MJdkt^3}{\epsilon}\right)^{\frac{1}{2+\theta}}+\frac{L(\Delta+JdkL)t^{\frac32}}{\epsilon^{\frac12}}+\frac{t^{\frac54}L(JdkM)^{\frac12}}{\epsilon^{\frac14}}\right),
\end{align}
we can ensure that the symmetry protection approach provides an approximation within error $\epsilon$ in the low-energy subspace below $\Delta$. {Here, the value of $0\leq\theta\leq1$ is a constant depending on the structure of the Hamiltonian~\cite{tran2021faster}}. 
\end{theorem}

The proof for \Cref{thm:SymProtLow} is provided in \Cref{sec:Sym}. The above theorem provides the step complexity for the symmetry protection approach in the low-energy simulation. Compared with the complexity required for full Hilbert space simulations~\cite{tran2021faster}, the $\lambda_H=O(n)$ dependence is replaced by $L\Delta$ in two terms and $O(n^{1/2})$ in the other two terms, which enables an improvement concerning the $n$ dependence. When $\theta>0$, the required Trotter number is smaller than that for implementing the standard first-order Lie-Trotter product formula in the low-energy subspace.

\paragraph{Power-law model.}
\Cref{thm:qDRIFTLow}, \Cref{thm:RandPermLow}, and \Cref{thm:SymProtLow} consider different digital quantum simulation approaches in simulating the low-energy quantum states of $k$-local spin models. We also consider power-law models, which can be regarded as a specific extension of $k$-local spin models at $k=2$. {In this case, there is no bound on the degree $d$ compared to the $k$-local spin model}, and additional constraints on the interaction strength $J$ for each term. We consider a $D$-dimensional power-law interaction model consisting of $n$ qubits. A power-law interaction $H=\sum_{\bi,\bj\in\Sigma}H_{\bi,\bj}$ with exponent satisfies:
\begin{align*}
\norm{H_{\bi,\bj}}\leq\begin{cases}
1 & \quad\text{if }\bi=\bj,\\
\frac{1}{\norm{\bi-\bj}_2^\alpha} & \quad\text{if }\bi\neq\bj,
\end{cases}
\end{align*}
where $\bi,\bj\in\Sigma$ are the qubit sites, $\Lambda\subseteq\R^D$ is a $D$-dimensional square lattice, and $H_{\bi,\bj}$ is an operator supported on two sites $\bi,\bj$. Some notable examples include Coulomb interactions ($\alpha=1$), dipole-dipole interactions ($\alpha=3$), and van der Waals interactions ($\alpha=6$). According to Lemma H.1 of Ref.~\cite{childs2021theory}, we have
\begin{align*}
g=\norm{\max_{\bi}\sum_{\bj\neq\bi} H_{\bi,\bj}}\leq\begin{cases}
O(n^{1-\alpha/D})&\quad\text{for }0\leq\alpha<D,\\
O(\log(n))&\quad\text{for }\alpha=D,\\
O(1)&\quad\text{for }\alpha>D.\\
\end{cases}
\end{align*}

Here, $g$ is an upper bound of the strengths of the interactions associated with a single spin qubit. We can observe that every term of the power-law Hamiltonian is $2$-local. In the following, we assume that each part $H_l$ {contains at most $M$ interaction terms}. We obtain the following complexity result for simulating such a power-law model in the low-energy subspace. 
\begin{theorem}\label{thm:PowerLow}
Consider simulating a $D$-dimensional power-law Hamiltonian $H=\sum_{l=1}^LH_l$ with interaction strength $g$ in the low-energy subspace $\Delta$. Assume that each $H_l$ {contains at most $M$ interaction terms}. Then, the Trotter number required to ensure the simulation error below $\epsilon$ is
\begin{align}\label{eq:PowerLow}
r=\tO\left(\frac{t^{1+\frac1p}L^{\frac1p}}{\epsilon^{\frac1p}}(\Delta+gq\log q)^{1+\frac 1p}\right)+O\left(\frac{t^{1+\frac{1}{2p+1}}}{\epsilon^{\frac{1}{2p+1}}}(LMg)^{\frac12+\frac{1}{4p+2}}\right).
\end{align}
The corresponding gate complexity is given as:
\begin{align*}
G=\tO\left(\frac{t^{1+\frac1p}L^{\frac1p+1}}{\epsilon^{\frac1p}}(\Delta+gq\log q)^{1+\frac 1p}\right)+O\left(\frac{Lt^{1+\frac{1}{2p+1}}}{\epsilon^{\frac{1}{2p+1}}}(LMg)^{\frac12+\frac{1}{4p+2}}\right).
\end{align*}
\end{theorem}

The proof for the above theorem is provided in \Cref{sec:Power}. Compared with the best previously known complexity result for simulating the power-law model~\cite{childs2021theory}, the step complexity in \eqref{eq:PowerLow} shows an improvement concerning the $n$ dependence at $p=1$. Notice that the best result for full Hilbert space simulations is obtained by simulating each term $H_{\bi,\bj}$ separately using one term $H_l$. For $p>1$, we can also obtain some advantage for specific choices of parameters $M$ and $\Delta$ in this case as we assume that $H_l$'s consist of more than one term of $H_{\bi,\bj}$.

\subsubsection{Lower bound for low-energy simulations}
We now consider the necessary number of queries to the Hamiltonian $H$ to simulate any given state $\rho$ in the low-energy subspace $\leq\Delta$. In \Cref{sec:lower}, we show the following theorem as the lower bound on the quantum query complexity of simulations in the low-energy subspace:
\begin{theorem}[Informal, see~\Cref{thm:lower} for the formal version]\label{thm:lower_informal}
There exists a Hamiltonian $H$ such that simulating some state with constant error within the low-energy threshold $\leq\Delta$ for some chosen scaled time $\tau$ requires at least
$\Omega\left(\max\{\tau, \frac{\log(\Delta)}{\log\log(\Delta)}\}\right)$ queries to $H$.
\end{theorem}
The above theorem indicates that, similar to full space simulations~\cite{berry2007efficient}, simulating Hamiltonians in the low-energy subspace also requires a linear number of queries to $H$ in the simulation time $\tau$. In addition, we also obtain a logarithmic dependence on the threshold $\Delta$. Although this is far from the upper bound, where we require $\poly(\Delta)$ queries, we remark that this gap also exists in the full space simulation~\cite{berry2014exponential}. The main idea to prove this theorem is to show a Hamiltonian of which exact low-energy subspace simulations for any time $\tau>0$ enable one to compute the parity of a string, which requires a linear number of queries to the length of the string~\cite{beals2001quantum,farhi1998limit}.

\subsection{Open questions}\label{sec:Outlook}
Our paper leaves several open questions for future investigations:
\begin{itemize}
    \item Ref.~\cite{zhao2022hamiltonian} developed a theory of average error for Hamiltonian simulation with random inputs and showed that improvement is possible when considering the average-case instead of the worst-case error. It is interesting to explore if we can obtain further improvement if we consider the simulation error for low-energy random states.
    \item In this paper, we study the simulation of evolution $V$ on a low-energy quantum state $\rho$. It is also meaningful to study the complexity for estimating the expected value $\Tr(O{V\rho V^\dagger})$ for some observable $O$ instead of the full state {$V\rho V^\dagger$} for a low-energy state $\rho$.
    \item Can we propose some Hamiltonians that are widely considered in the field of quantum information such that we can obtain improvement in the step complexity of digital quantum simulation under the low-energy assumption?
    
    \item {Throughout the paper, the norm of the Hamiltonian is assumed to be bounded and its dimension is assumed to be finite. A natural open question is whether we can generalize our results concerning simulations in low-energy subspace to Hamiltonians with unbounded norm or infinite-dimensional Hilbert space.}
\end{itemize}

\paragraph{Roadmap.}
The rest of the paper is organized as follows. In \Cref{sec:Frame}, we provide the general framework for low-energy simulations. In particular, we recap some of the results in \cite{csahinouglu2021hamiltonian} and prove several important lemmas for our results. We clarify the difference between our work and previous works. We further prove the robustness of the improvement originates from the low-energy assumption in imperfect state preparation {due to thermalization}. In \Cref{sec:Appl}, we consider the applications of low-energy settings for different digital quantum simulation approaches and quantum models. In \Cref{sec:Rand}, we consider randomized approaches including qDRIFT and random permutation for simulating $k$-local Hamiltonians, and provide the proofs for \Cref{thm:qDRIFTLow}, \Cref{thm:RandPermLow}, and \Cref{coro:DoubLow}. In \Cref{sec:Sym}, we consider the symmetry protection approach and provide the proof for \Cref{thm:SymProtLow}. We show that the low-energy setting can provide improvements in power-law Hamiltonians in \Cref{sec:Power} and prove \Cref{thm:PowerLow}. Finally, we prove the query lower bound in \Cref{thm:lower_informal} for simulating dynamics in the low-energy subspace in \Cref{sec:lower}.


\section{The General Framework}\label{sec:Frame}

\paragraph{Notations. }We consider simulating the Hamiltonian $H=\sum_{l=1}^LH_l$ composed of $n$ particles. We assume that each term $H_l$ contains at most $M$ interaction terms and each interaction term acts on at most $k$ qubits. For each qubit, we assume that the strength of interaction between this qubit and the rest qubits is bounded by $g$ measured by the spectral norm. Without loss of generality, we assume that $H_l\geq 0$ for any $l$. We also assume the number of terms acting on each qubit is bounded by $d$ and the strength of each term is bounded by $J$. Thus, we have $g\leq dJ$. 

Throughout this paper, we use the notation $\tO(\cdot)$ which omits the polylogarithmic dependence on the parameters. Given positive constant $\Lambda$, we denote $\pll$ the projector to the subspace spanned by eigenstates of $H$ with energy smaller than $\Lambda$. We further denote $\plg=I-\pll$ as the projector to the orthogonal subspace. {We also summarize the notations used throughout this paper in Table~\ref{tab:notations}.}
\begin{table}[ht]
\centering
\resizebox{1.0\columnwidth}{!}{
{\begin{tabular}{ll|ll}
\hline
Symbol & Definition & Symbol & Definition \\ 
\hline
\hline
$H$ & Hamiltonian & $L$& $H=\sum_{l=1}^L H_l$ terms \\
$A$ & Operator to apply projection & $R_A$ & Interaction terms strength on $A$ \\ 
$M$ & Interaction terms number in $H_l$ & $J$ & Interaction term strength\\
$\Delta$ & Low-energy parameter & $\Delta'$ & $\geq\Delta$, effective low-energy norm \\
$d$ & Degree of Hamiltonian & $k$ & $k$-local Hamiltonian \\
$g$ & Interaction strength on single qubit & $\lambda$ &$(2gk)^{-1}$, abbreviation\\
$n$ & Number of particles & $t$ & Total evolution time\\
$\epsilon$ & Target simulation error & $r$ & Trotter number \\
$\Pi_{\leq\Lambda}$ & Low-energy subspace projector & $\delta$ & $t/r$, unit time step \\
$p$ & Order of product formula & $q$ & Terms number of product formula\\
$\lambda_h$ & $\max_{l}\norm{H_l},~\lambda_{\overline{h}}=\max_{l}\norm{\oH_l}$ & $\lambda_H$ & $\sum_{l=1}^L\norm{H_l},~\lambda_{\oH}=\sum_{l=1}^L\norm{\oH_l}$\\
$C_\mu$ & Symmetry transformation operator & $B_k$ & Matrix martingale, $C_k=B_k-B_{k-1}$\\
$\theta$ & Symmetry transformation parameter & $\chi$ & Failure probability\\
$V$ & Ideal time-evolution operator & $U$ & Unitary operator to approximate $V$\\ 
$G$ & Gate complexity & $S$ & Suzuki-Trotter formula, $W=S_p$ \\
\hline
\end{tabular}}
}
\caption{{The notation table.}}
\label{tab:notations}
\end{table}

\paragraph{Simulating $k$-local Hamiltonians in the low-energy subspace. } For $k$-local Hamiltonians with bounded interaction strengths, the following lemma holds as the backbone of the framework for analyzing the simulation performance in the low-energy subspace.
\begin{lemma}[Theorem 2.1 of~\cite{arad2016connecting}]\label{lem:ProjBound}
Given $H=\sum_{l=1}^LH_l$ defined above with parameters $g$, $L$, and $k$, and any operator $A$, the following inequality holds
\begin{align*}
\norm{\plpg A\pll}\leq\norm{A}\cdot e^{-\lambda(\Lambda^\prime-\Lambda-2R_A)},
\end{align*}
where $R_A$ is the strengths of interaction terms acting on $A$ and $\lambda=(2gk)^{-1}$. Here, $\Lambda'\geq\Lambda\geq 0$ are two positive values.
\end{lemma}

Based on \Cref{lem:ProjBound}, we can obtain some corollaries when $A$ is an evolution of some terms $H_l$. We list these lemmas in \Cref{app:Lemma}. Using these lemmas, Ref.~\cite{csahinouglu2021hamiltonian} obtained the following result concerning the performance of arbitrary product formulas of length $q$ in the low-energy subspace. Given $H=\sum_{l=1}^LH_l$ an $n$-qubit $k$-local Hamiltonian with parameters $M$, $J$, and $d$ ($LMJ=O(n)$), the simulation error between the target evolution $V(\delta)$ and a product formula $W(\delta)$ of $p$-th order error is bounded by~\cite{csahinouglu2021hamiltonian}:
\begin{align}\label{eq:SimErrRep}
\norm{(V(\delta)-W(\delta))\pdl}=O((L\Delta'\delta)^{p+1}),
\end{align}
where 
\begin{align}\label{eq:SimErrRepDelta}
\Delta'=\Delta+\frac{\alpha LM\delta}{\lambda}+\frac{q\log q}{\lambda}+\frac{q}{\lambda}\log\left(\frac{1}{J\delta}\right),
\end{align}
$\alpha$ is some constant, and $\lambda=(2Jdk)^{-1}$. We can further deduce the following lemma regarding the Trotter number required for low-energy simulations based on the above error decomposition.
\begin{lemma}[Eq.~(111) of~\cite{csahinouglu2021hamiltonian}]\label{lem:ProdLow}
Let $H=\sum_{l=1}^LH_l$ be an $n$-qubit $k$-local Hamiltonian with parameters $M$, $J$, and $d$ ($LMJ=O(n)$). By choosing the Trotter number
\begin{align}\label{eq:ProdLow}
r=\frac{t}{\delta}=\tO\left(\frac{t^{1+\frac1p}}{\epsilon^{\frac1p}}(L\Delta+LdkJq\log q)^{1+\frac1p}\right)+O\left(\frac{t^{1+\frac{1}{2p+1}}}{\epsilon^{\frac{1}{2p+1}}}(L^2dMJ^2)^{\frac12+\frac{1}{4p+2}}\right),
\end{align}
we can ensure that the $p$-th order product formula $W(\delta)$ provides an approximation within error $\epsilon$ in the low-energy subspace below $\Delta$.
\end{lemma}

Furthermore, as the implementation of the $p$-th order Suzuki-Trotter formula requires {$O(2\cdot5^{p/2-1}L)$ gates~\cite{suzuki1990fractal,berry2007efficient}}, the gate complexity required is
\begin{align*}
G=\tO\left(\frac{Lt^{1+\frac1p}}{\epsilon^{\frac1p}}(L\Delta+LdkJq\log q)^{1+\frac1p}\right)+O\left(\frac{Lt^{1+\frac{1}{2p+1}}}{\epsilon^{\frac{1}{2p+1}}}(L^2dMJ^2)^{\frac12+\frac{1}{4p+2}}\right).
\end{align*}

Notice that $LMJ=O(n)$, the step complexity has an explicit $n^{\frac12\left(1+\frac{1}{2p+1}\right)}$ term. For $p=1$, the step complexity obtained by \Cref{lem:ProdLow} achieves a reduction concerning the dependence on $n$ compared to the performance for simulating high energy states. For general $p\geq1$ and {$d=O(n^{\frac{1}{2(p+1)}})$}, we can also obtained a similar improvement. 

As an illustration for the reduction of simulating dynamics of Hamiltonian assuming input states from the low-energy subspace, we perform numerical experiments to benchmark the product formulas in simulating $k$-local Hamiltonians. In particular, we study the following  homogeneous Heisenberg model without external fields: 
\begin{equation}\label{eq:Heisenberg}
    H = -\sum_{<i,j>}X_iX_j+Y_iY_j+Z_iZ_j,
\end{equation}
where $X_i, Y_i$, and $Z_i$ are Pauli operators acting on the $i$-th spin, and the summation is over every adjacent spin pair. Since the ground energy $E_{0l}<0$ for every $H_l$, we implement the shift $H_l\rightarrow H_l-E_{0l}I$ for every $H_l$. We empirically pick the threshold $\Delta = 4$ to make sure that the corresponding low-energy subspace is neither (close to) the full Hilbert space nor an empty subspace.

\begin{figure}[ht]
    \centering
    \begin{subfigure}[ht]{.2\linewidth}
        \includegraphics[width=\linewidth]{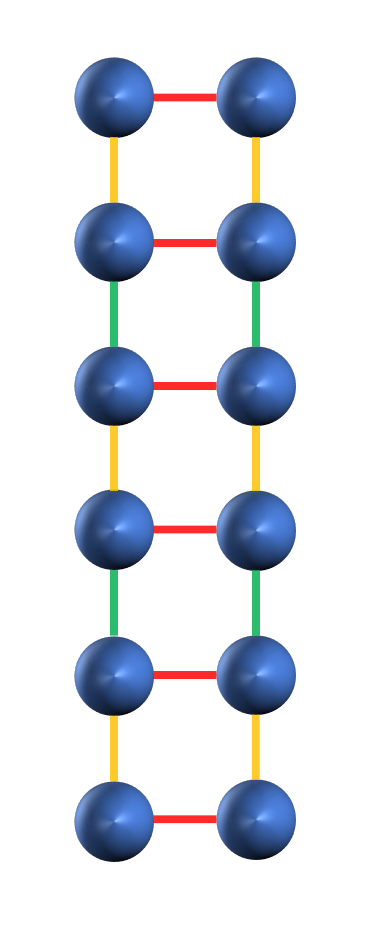}
    \end{subfigure}
    \begin{subfigure}[ht]{.75\linewidth}
        \includegraphics[width=\linewidth]{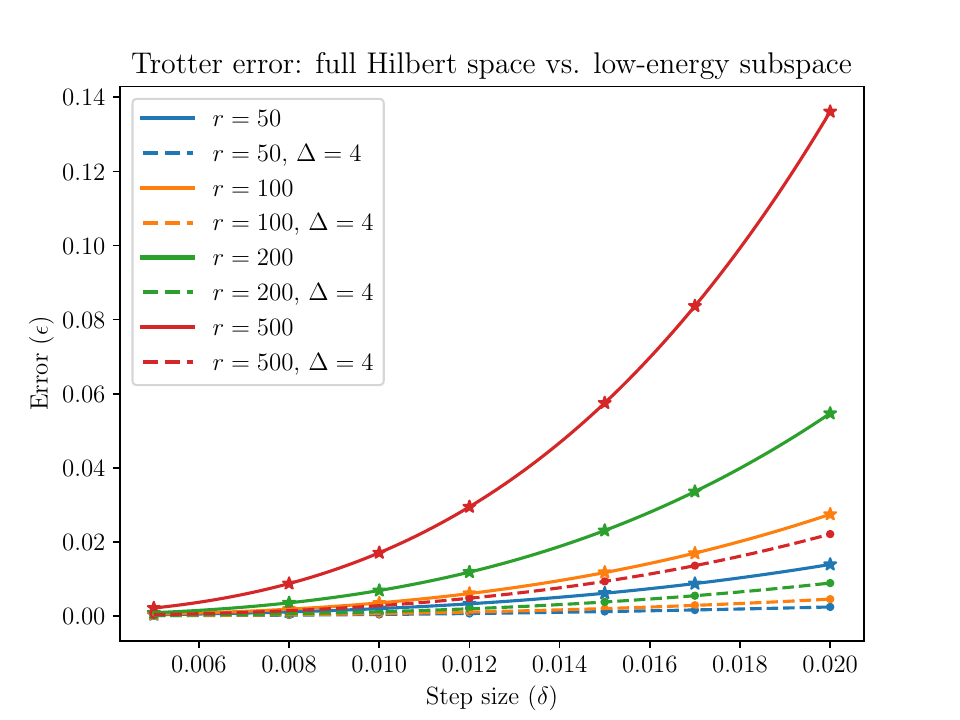}
    \end{subfigure}
    \caption{\textbf{Quantum simulation in full Hilbert space vs. low-energy subspace: Trotter error induced by product formulas for a $\mathbf{2\times 6}$ homogeneous Heisenberg spin-$\mathbf{1/2}$ model.} The solid lines denote simulation errors in the full Hilbert space while the dashed lines denote errors in the low-energy subspace. The plot distinguishes different Trotter numbers $r$ by the color of the lines.}
    \label{fig:trotter-2x6}
\end{figure}

As shown in the left-hand side of \Cref{fig:trotter-2x6}, the Hamiltonian in \eqref{eq:Heisenberg} is decomposed into three interaction terms: $H = H_1+H_2+H_3$, respectively represented by yellow, green, and red bonds. Here, a bond represents all the interactions between a pair of spins. In the right-hand side of \Cref{fig:trotter-2x6}, we plot the simulation {error $\epsilon$} as a function of single-step evolution time $\delta$ with respect to different Trotter numbers $r$. The time evolution operator $V(\delta)$ is approximated by the second-order product formula $S_2(\delta)$. 
We can observe that the projection to the low-energy subspace significantly reduces the error compared to the full Hilbert space for the same parameter configuration. {We further remark that the scaling of the numerical simulation regarding parameters fits perfectly with the theoretical results.}

\subsection{Low-energy simulation with imperfect state preparations {due to thermalization}}\label{sec:ImperfectPrep}
In the previous analysis, we focused on simulating the dynamics of low-energy input states for a given Hamiltonian. We assumed the quantum state is chosen in the low-energy subspace under a threshold $\Delta$. However, the state preparation on near-term devices is not perfect, resulting in a Gaussian tail distribution over higher energy states on the energy spectrum. It is natural to ask if the improvement obtained by low-energy simulations persists under the imperfect state preparation scenario. In the following, we provide an affirmative answer to this question.

{In real quantum systems and experiments, the actually prepared states usually deviate from the ideal states due to thermalization~\cite{banuls2020entanglement,ge2019faster,lu2021algorithms}, which results in a Gaussian tail on the spectrum distribution beyond the energy threshold $\Delta$}. We focus on pure states. and suppose the perfect input state $\ket{\psi_p}$ has energy $\Delta$, and denote the actual input state to be $\ket{\psi}$. After thermalization deviation, the mean energy and variance of $H$ in the state $\ket{\psi}$ is
\begin{align*}
E_\psi&=\bra{\psi}H\ket{\psi},\\
\sigma_\psi^2&=\bra{\psi}(H-{E_\psi})^2\ket{\psi}.
\end{align*}

In the case when $H$ is local and $\ket{\psi}$ has finite correlation length, $\sigma_\psi^2$ will scale as $O(n)$~\cite{banuls2020entanglement}. Assuming imperfect state preparation, we adapt {$E_\psi=\Delta$} and $\sigma_{\psi}^2\ll \Delta$ from current analog state preparation schemes~\cite{ge2019faster,lu2021algorithms}. Quantitatively, the support of the imperfect input state is a Gaussian distribution with parameter $\sigma_\psi$ around $\Delta$. By using the error representation in \eqref{eq:SimErrRep} and \eqref{eq:SimErrRepDelta}, we can compute the simulation error $\text{Err}_{\text{imp}}$ as:
\begin{align}\label{eq:EqqImp}
\text{Err}_{\text{imp}}&\propto\int_{-\infty}^{\infty}e^{-(x-\Delta)^2/\sigma_\psi^2}(L\Delta_x\delta)^{p+1}dx,
\end{align}
where
\begin{align*}
\Delta_x=x+\frac{\alpha LM\delta}{\lambda}+\frac{q\log q}{\lambda}+\frac{q}{\lambda}\log\left(\frac{1}{J\delta}\right).
\end{align*}

Recall that for standard Suzuki-Trotter formulas {obtained by the Yoshida method~\cite{Yoshida1990construcation}}, $p$ is an even number. Therefore, we can compute $\text{Err}_{\text{imp}}$ as
\begin{align*}
\text{Err}_{\text{imp}}\propto(L\delta\sigma_\psi)^{p+1}\left[\left(\frac{\Delta'}{\sigma_\psi}\right)^{(p+1)}+\sum_{i=1}^{p/2}\binom{p+1}{2i}\left(\frac{\Delta'}{\sigma_\psi}\right)^{(p+1-2i)}\cdot\prod_{j=1}^{i}\frac{1}{2(2j-1)}\right].
\end{align*}

As $\sigma_\psi\ll \Delta\leq \Delta'$ and $p$ a constant, we can discard the lower-order terms for $\Delta'$ and obtain
\begin{align*}
\text{Err}_{\text{imp}}\simeq(L\delta\Delta')^{p+1}\left(1+\frac{p(p+1)}{4}\left(\frac{\sigma_\psi}{\Delta'}\right)^2\right).
\end{align*}

We compare the above error decomposition with \eqref{eq:SimErrRep} and observe that we only get an additional $O(\Delta^{'(p-1)}\sigma_\psi^2)$ term when the state preparation is imperfect. Therefore, we conclude that for imperfect state preparation {due to thermalization}, the improvement obtained by the low-energy simulation assumption persists. We further claim that this robustness persists for all the analyses on \textit{any} following low-energy performance for various digital quantum simulation approaches. This is because the low-energy simulation error for any following algorithms has a polynomial scaling on $\Delta'$. We can thus employ similar calculations as \eqref{eq:EqqImp} and prove the robustness.


\section{Applications}\label{sec:Appl}
Intuitively, the resource requirement for different simulations reduces when the evolved quantum system does not consider high-energy states. Yet, \Cref{lem:ProdLow} only provides strict proof for standard product formulas and $k$-local Hamiltonians. In this work, we explore other digital quantum simulation algorithms and Hamiltonians. In this section, we provide the technical details for our results in applying low-energy analysis for different digital quantum simulation algorithms and Hamiltonians.

\subsection{Randomized product formulas}\label{sec:Rand}
In this subsection, we analyze the performance of randomized simulation approaches in the low-energy subspace. In particular, we provide the proof for \Cref{thm:qDRIFTLow}, \Cref{thm:RandPermLow}, and \Cref{coro:DoubLow}.
\vspace{-.5em}
\paragraph{qDRIFT algorithm. }We begin with recapping the qDRIFT algorithm in the below \Cref{algo:qDRIFT}.

\begin{algorithm}[ht]
\caption{The qDRIFT algorithm}
\label{algo:qDRIFT}
\begin{algorithmic}[1]
\REQUIRE Hamiltonian $H=\sum_{l=1}^LH_l$, parameter $\lambda_{\oH}=\sum_l\norm{\oH_l}$ where $\overline{H}_l=\pdpl H_l\pdpl$, evolution time $t$, and number of steps $r$
\STATE At each time interval $t/r$: evolve a random term in
Hamiltonian $U_i=e^{-i(t/r) X_i}$, where $X_i$ is randomly chosen as $\lambda_{\oH} H_l/\norm{\oH_l}$ with probability $\norm{\oH_l}/\lambda_{\oH}$
\ENSURE The unstructured (randomly generated)
product formula $U=U_r\ldots U_1$
\end{algorithmic}
\end{algorithm}

We now analyze the performance of the qDRIFT algorithm in the low-energy subspace and prove \Cref{thm:qDRIFTLow}. We still consider the $k$-local Hamiltonians and parameters $d$, $M$, {and $J$} defined above. We carefully pick $\Delta'\geq\Delta$ a value to be determined later. We denote $\oH=\pdpl H\pdpl$ and $\oH_l=\pdpl H_l\pdpl$. Concerning $H_l\geq 0$, we have $\bra{\psi}\pdpl H_l\pdpl\ket{\psi}\leq \sum_{l=1}^L \bra{\psi}\oH_l\ket{\psi}=\bra{\psi}\oH\ket{\psi}\leq\norm{\oH},~\forall \ket{\psi}$. Therefore, $\norm{\oH_l}\leq\norm{\oH}\leq\Delta'$.

Given a random sequence $U_rU_{r-1}\ldots U_1$ generated by \Cref{algo:qDRIFT}, our goal is to estimate the error
\begin{align*}
\norm{(U_r\ldots U_1-V)\pdl},
\end{align*}
where $V=e^{-iHt}$ is the ideal evolution. To begin with, we decompose this error into three terms:
\begin{align*}
&\quad\ \norm{(U_r\ldots U_1-V)\pdl}\nonumber\\
&\leq\norm{(U_r\ldots U_1-\overline{U}_r\ldots\overline{U}_1)\pdl}+\norm{(\overline{U}_r\ldots\overline{U}_1-{(\E[\overline{U}_i])^r})\pdl}+\norm{({(\E[\overline{U}_i])^r}-V)\pdl}\nonumber\\
&\leq\norm{(U_r\ldots U_1-\overline{U}_r\ldots\overline{U}_1)\pdl}+\norm{\overline{U}_r\ldots\overline{U}_1-{(\E[\overline{U}_i])^r}}+\norm{{(\E[\overline{U}_i])^r}-\overline{V}},
\end{align*}
where $\overline{V}=e^{-i\oH t}$ and $\overline{U}_i$ are obtained by running \Cref{algo:qDRIFT} on $\oH$ with corresponding $\oH_l$'s, and choosing $\overline{X}_i=\lambda_{\oH}\oH_{l}/\norm{\oH_{l}}$ in the $i$-th  step with probability $\norm{\oH_l}/\lambda_{\oH}$. Here, the first term is the projection error for the low-energy subspace estimation. The second term is the random fluctuation in the low-energy subspace. The third term is the deterministic bias in the low-energy subspace. 

We first consider the deterministic bias error term. For the projected Hamiltonian $\oH$, $\lambda_{\oH}=\sum_{l=1}^L\norm{\oH_l}=O(L\Delta')$ for some $\Delta'\geq\Delta$ to be fixed later. It is worthwhile to mention that for $k$-local Hamiltonians, $\lambda_{\oH}$ is explicitly independent of system size $n$ while $\lambda_H=O(n)$~\cite{childs2021theory}. According to Lemma 3.5 of Ref.~\cite{chen2021concentration}, for Hamiltonian $\oH=\sum_{l=1}^L{\oH_l}$ and $\overline{U}_i$, the deterministic bias term is bounded by 
\begin{align*}
\norm{{(\E[\overline{U}_i])^r}-\overline{V}}\leq r\norm{\E[\overline{U}_i]-\overline{V}^{1/r}}\leq O\left(r\cdot\frac{\lambda^{2}_{\oH} t^2}{r^2}\right)=O\left(\frac{\lambda^{2}_{\oH} t^2}{r}\right).
\end{align*}

Here $\overline{V}^{1/r}=e^{-iHt/r}=e^{iH\delta}$.

The next step is to provide a bound for the random fluctuation term $\norm{\overline{U}_r\ldots\overline{U}_1-{(\E[\overline{U}_i])^r}}$. For the randomized sequence generated by the qDRIFT algorithm projected into the low-energy subspace $\overline{U}_r\ldots\overline{U}_1$, we denote $B_k=(\E[\overline{U}_i])^{r-k}\overline{U}_k\cdots \overline{U}_1$. Therefore, $B_r=\overline{U}_r\ldots\overline{U}_1$ and $B_0=(\E[\overline{U}_i])^r$. We can verify that this sequence $\{B_k:k=0,\ldots,r\}$ satisfies the martingale properties as follows:
\begin{itemize}
\item \textit{Causality: } every $B_k$ completely depends on the previous information of $B_{k-1},\ldots,B_1$, i.e., the choice of $\overline{U}_{k-1},\ldots,\overline{U}_1$.\vspace{-.1em}
\item \textit{Status quo: } the expectation value for $B_k$ equals to $B_{k-1}$ conditioned on the previous history, i.e., $\E[B_k|B_{k-1},\ldots,B_1]=B_{k-1}$ for all $k=1,\ldots r$.
\end{itemize}

The property for such martingale features similar properties to an unbiased random walk. Based on Freedman's inequality and its application to martingales~\cite{huang2022matrix,oliveira2009spectrum,tropp2011freedman,pinelis1994optimum,kathuria2020concentration,tropp2015introduction}, we have \Cref{lem:MartConc} in \Cref{app:Lemma}. Starting from this lemma, we consider the martingale provided by the qDRIFT algorithm defined above. We define $C_k=B_k-B_{k-1}$. We observe that
\begin{align*}
\norm{C_k}&=\norm{B_k-B_{k-1}}\nonumber\\
&=\norm{\E[\overline{U}_i]^{r-k}(\overline{U}_{k}-\E[\overline{U}_i])\overline{U}_{k-1}\ldots \overline{U}_1}\nonumber\\
&\leq{(\norm{\E[\overline{U}_i]})^{r-k}}\norm{(\overline{U}_{k}-\E[\overline{U}_i])}\norm{\overline{U}_{k-1}\ldots \overline{U}_1}\nonumber\\
&{\leq(\E[\norm{\overline{U}_i}])^{r-k}}\norm{(\overline{U}_{k}-\E[\overline{U}_i])}\nonumber\\
&=2\delta\lambda_{\oH},
\end{align*}
{where the second inequality follows the convexity of mathematical expectation to arrive at the fourth line.}

We invoke \Cref{lem:MartConc} with $v$ bounded by $r\norm{C_k}^2$ and obtain that
\begin{align*}
\Pr[\norm{B_0-B_r}\geq\epsilon]\leq {2^{n+1}}\exp\left(-\frac{3r\epsilon^2}{24(t\lambda_{\oH})^2+4(t\lambda_{\oH})\epsilon}\right),
\end{align*}
where $2^n$ is the dimension of the Hilbert space. With probability at least $1-\chi$, we have
\begin{align}\label{eq:qDRIFT-dependence}
\norm{\overline{U}_r\ldots\overline{U}_1-{(\E[\overline{U}_i])^r}}\leq \tO\left(\sqrt{\frac{\lambda_{\oH}^2t^2}{r}\left(n+\log\left(\frac{1}{\chi}\right)\right)}\right).
\end{align}

Finally, we provide the bound for the projection error $\norm{(U_r\ldots U_1-\overline{U}_r\ldots\overline{U}_1)\pdl}$. {Intuitively, we have
\begin{align*}
&\norm{(U_r\ldots U_1-\overline{U}_r\ldots\overline{U}_1)\pdl}\leq\sum_{i=1}^r\norm{(U_i-\overline{U}_i)\pdpl}.
\end{align*}
This is because $[\overline{U}_i,\pdpl]=0$ for any $i=1,\ldots,r$ and $\pdpl\pdpl=\pdpl$. However, this yields no bound on this term $\norm{(U_r\ldots U_1-\overline{U}_r\ldots\overline{U}_1)\pdl}$. To address this issue, we consider the following two-phase decomposition of a qDRIFT algorithm. We write $r=r_1r_2$ with $r_1$ and $r_2$ to be fixed later. We decompose the qDRIFT sequence of $r$ steps into $r_1$ groups of ``random product formulas", each of $r_2$ steps. The full qDRIFT sequence can be written as $W_{r_1}\ldots W_1$ where $W_{i}=U_{i,r_2}\ldots U_{i,1}$ for $i=1,\ldots,r_1$ with $U_{i,j}$ the random qDRIFT step for $j=1,\ldots,r_2$. We thus decompose the simulation error as
\begin{align*}
\norm{(V-W_{r_1}\ldots W_1)\pdl}&=\norm{\left(\left(V^{1/r_1}\right)^{r_1}-W_{r_1}\ldots W_1\right)\pdl}\\
&\leq \sum_{i=1}^{r_1}\norm{(V^{1/r_1}-W_i)\pdl}\\
&\leq r_1\max_i\norm{(V^{1/r_1}-W_i)\pdl},
\end{align*}
where $V=e^{-iHt}$ the second line follows from the fact that $[V,\pdl]=[V^{1/r_1},\pdl]=0$. We then consider an arbitrary $W_i$ (denoted as $W$ in the following), and employ the above error decomposition as:
\begin{align*}
\norm{(V^{1/r_1}-W)\pdl}&=\norm{(V^{1/r_1}-U_{r_2}\ldots U_1)\pdl}\\
&\leq\norm{(U_{r_2}\ldots U_1-\overline{U}_{r_2}\ldots\overline{U}_1)\pdl}+\norm{\overline{U}_{r_2}\ldots\overline{U}_1-(\E[\overline{U}_i])^{r_2}}+\norm{(\E[\overline{U}_i])^{r_2}-\overline{V}^{1/r_1}}.
\end{align*}
According to the above calculation, with probability $1-\chi$ we have the second and the third terms as:
\begin{align*}
\norm{\overline{U}_{r_2}\ldots\overline{U}_1-(\E[\overline{U}_i])^{r_2}}&=\tO\left(\sqrt{\frac{\lambda_{\oH}^2(t/r_1)^2}{r_2}\left(n+\log\left(\frac{1}{\chi}\right)\right)}\right)\\
\norm{(\E[\overline{U}_i])^{r_2}-\overline{V}^{1/r_1}}&=O\left(\frac{\lambda^{2}_{\oH}(t/r_1)^2}{r_2}\right),
\end{align*}
where $\lambda_H=O(L\Delta')$.
}

We now provide a bound for the truncation error {$\norm{(U_{r_2}\ldots U_1-\overline{U}_{r_2}\ldots\overline{U}_1)\pdl}$}. We assume that $U_i$ evolves on $H_{i}$ for some $i=1,\ldots,r_2$ and denote $\delta_i'=\lambda_{\oH}\delta/\norm{\oH_i}$, {$\delta'=\max_i \delta_i'$}, {the error can be bounded as follows follows from \Cref{lem:LeakExp} and \Cref{lem:LeakDiff}:
\begin{align*}
\norm{(U_{r_2}\ldots U_1-\overline{U}_{r_2}\ldots\overline{U}_1)\pdl}\leq \exp\left(-\frac{1}{r_2}\left(\lambda(\Delta'-\Delta)-\alpha\abs{\bm{\delta'}}M-r_2\log r_2\right)\right),
\end{align*}
where $\alpha=eJ$, $\lambda=(2Jdk)^{-1}$, and $\bm{\delta'}=(\delta_{r_2}',\ldots,\delta_1')$ with $\abs{\bm{\delta'}}=O(Lt/r_1)$.
} 

Now, we are ready to prove \Cref{thm:qDRIFTLow}. We first consider the expectation error and prove \eqref{eq:qDRIFTLowExp}. In this case, we do not need to guarantee the random fluctuation term. {We compute the $r_2$ such that the projection error is of the same scale as the deterministic bias, i.e.,
\begin{align*}
\norm{(U_{r_2}\ldots U_1-\overline{U}_{r_2}\ldots\overline{U}_1)\pdl}\simeq \norm{{(\E[\overline{U}_i])^{r_2}}-\overline{V}^{1/r_1}}.
\end{align*}}

Based on the previous bound on these two terms of error, we conclude that
{\begin{align*}
\exp\left(-\frac{1}{r_2}\left(\lambda(\Delta'-\Delta)-\alpha\abs{\bm{\delta'}}M-r_2\log r_2\right)\right)=O\left(\frac{\lambda^{2}_{\oH}(t/r_1)^2}{r_2}\right).
\end{align*}}

To achieve this, we can fix the value of $\Delta'$ as follows:
{\begin{align*}
\Delta'=\Delta+O\left(\frac{1}{\lambda }\alpha \frac{t}{r_1} ML\right)+\frac{1}{\lambda}r_2\log r_2+\frac{1}{\lambda}r_2\log(\frac{r_2}{\lambda_{\oH}^2(t/r_1)^2}).
\end{align*}}

The total error is thus $O\left(\frac{\lambda_{\oH}^2 t^2}{r}\right)$ with $\lambda_{\oH}=O(L\Delta')$. Now, we consider three cases separately to derive the final step complexity. 

In the first case, $\Delta$ is the dominant term of $\Delta'$. In this case, we require
\begin{align*}
\frac{L^2\Delta^2t^2}{r}\leq\epsilon\to r\geq\frac{L^2\Delta^2t^2}{\epsilon}.
\end{align*}

In the second case, $O\left(\frac{1}{\lambda }\alpha \frac{t}{r_1} ML\right)$ is the dominant term of $\Delta'$. {We require
\begin{align*}
\frac{L^4\alpha^2 r_2^2M^2t^4}{\lambda^2r^3}\leq O(\epsilon)\to r\geq O( L^{4/3}M^{2/3}J^{4/3}r_2^{2/3}t^{4/3}\epsilon^{-1/3}).
\end{align*}}

{In the last case, $\frac{1}{\lambda}r_2\log r_2+\frac{1}{\lambda}r_2\log(\frac{r_2}{\lambda_{\oH}^2(t/r_1)^2})$ is the dominant term of $\Delta'$. Except for a mild polylogarithmic correction, this term is the same as $r_2\log r_2/\lambda$}. In this case, we can prove that the corresponding step complexity is not the dominant term. Optimizing over all choices of $r_2$, we have the step complexity minimized when $r_2=O(1)$.

In general, the Trotter step complexity is 
\begin{align*}
r=\tO\left(\frac{L^2(\Delta+dkJ)^2t^2}{\epsilon}\right)+O\left(\frac{(LJ)^{4/3}M^{2/3}t^{4/3}}{\epsilon^{1/3}}\right),
\end{align*}
which finishes the proof for \eqref{eq:qDRIFTLowExp}. The corresponding gate complexity is
\begin{align*}
G=\tO\left(\frac{L^3(\Delta+dkJ)^2t^2}{\epsilon}\right)+O\left(\frac{(L^{7/4}J)^{4/3}M^{2/3}t^{4/3}}{\epsilon^{1/3}}\right),
\end{align*}
given that $O(L)$ gates are required to implement a Suzuki-Trotter formula in each time step.

Now, we consider the step complexity required to ensure qDRIFT converges and prove \eqref{eq:qDRIFTLowProb}. If we want to further control the random fluctuation, {at $r_2=O(1)$} we require
\begin{align*}
\exp\left(-[\lambda(\Delta'-\Delta)-\alpha \delta' M]\right)=O\left(\sqrt{\frac{\lambda_{\oH}^2t^2}{r}\cdot\left(n+\log(\frac{1}{\chi})\right)}\right).
\end{align*}

To achieve this, we can fix the value of $\Delta'$ as follows:
\begin{align*}
\Delta'=\Delta+\frac{1}{\lambda}\alpha\delta' M+\frac{1}{2\lambda}\log(\frac{r^{3}}{\lambda_{\oH}^2t^2(n+\log(1/\chi))}).
\end{align*}

Following a similar procedure to the average-case performance, we derive the Trotter number to ensure a $1-\chi$ success probability for qDRIFT as
\begin{align}\label{eq:qDRIFTNum}
r= \tO\left(\frac{L^2(\Delta+dkJ)^2 t^2}{\epsilon^2}\left(n+\log(\frac{1}{\chi})\right)\right)+O\left(\frac{(LJ)^{4/3}M^{2/3}t^{4/3}}{\epsilon^{2/3}}\left(n+\log(\frac{1}{\chi})\right)^{1/3}\right),
\end{align}
which finishes the proof for \Cref{thm:qDRIFTLow}. In addition, the corresponding gate complexity is
\begin{align*}
G=\tO\left(\frac{L^3(\Delta+dkJ)^2 t^2}{\epsilon^2}\left(n+\log(\frac{1}{\chi})\right)\right)+O\left(\frac{(L^{7/4}J)^{4/3}M^{2/3}t^{4/3}}{\epsilon^{2/3}}\left(n+\log(\frac{1}{\chi})\right)^{1/3}\right).
\end{align*}

\paragraph{Random permutation. }In the previous qDRIFT algorithm, we only consider applying a random evolution in each time step, which is a random version of the first-order Lie-Trotter formula. A randomized approach for higher-order formulas, known as random permutation, was proposed in Ref.~\cite{childs2019faster}. As shown in \Cref{algo:RandPerm}, this algorithm randomly permutes the order of Hamiltonian terms within each block to $S_{p}^{\sigma_i}(\delta=t/r)$ for a permutation $\sigma_i\in\cS_L$ from the permutation group.
 
\begin{algorithm}[ht]
\caption{The random permutation algorithm}
\label{algo:RandPerm}
\begin{algorithmic}[1]
\REQUIRE Hamiltonian $H=\sum_{l=1}^LH_l$, evolution time $t$, and number of steps $r$
\STATE At each time interval $\delta=t/r$: evolve a random term in
Hamiltonian $U_i=S_{p}^{\sigma_i}(\delta)$ with a randomly chosen sequence with probability $1/L!$
\ENSURE The unstructured (randomly generated)
product formula $U=U_r\ldots U_1$
\end{algorithmic}
\end{algorithm}

We now consider the performance of this algorithm in the low-energy subspace and prove \Cref{thm:RandPermLow}. Similar to the previous section, we can decompose the error into three terms:
\begin{align*}
&\quad\ \norm{(U_r\ldots U_1-V)\pdl}\nonumber\\
&\leq\norm{(U_r\ldots U_1-\overline{U}_r\ldots\overline{U}_1)\pdl}+\norm{(\overline{U}_r\ldots\overline{U}_1-{(\E[\overline{U}_i])^r})\pdl}+\norm{({(\E[\overline{U}_i])^r}-V)\pdl}\nonumber\\
&\leq\norm{(U_r\ldots U_1-\overline{U}_r\ldots\overline{U}_1)\pdl}+\norm{\overline{U}_r\ldots\overline{U}_1-{(\E[\overline{U}_i])^r}}+\norm{{(\E[\overline{U}_i])^r}-\overline{V}},
\end{align*}
where $\overline{U}_i=S_{p}^{\sigma_i}(\delta)$ with Hamiltonian $\oH$ and $\overline{U}=e^{-i\oH t}$. The first term is the projection error for the low-energy subspace estimation. The second term is the random fluctuation in the low-energy subspace. The third term is the deterministic bias in the low-energy subspace. We introduce two new parameters $\lambda_h=\max_l\norm{H_l}=JML=O(n)$ and $\lambda_{\oh}=\max_l\norm{\oH_l}=\Delta'$ for $n$-qubit $k$-local Hamiltonian.

We begin with the deterministic bias error term. According to the direct calculation~\cite{childs2019faster,suzuki1985decomposition}, we can derive that (see, e.g. Eq. (C14) of Ref.~\cite{chen2021concentration})
\begin{align*}
\norm{\E[\overline{U}_i]-V^{1/r}}\leq O\left(\left(\frac{t\lambda_{\oh}}{r}\right)^{p+1}L^{p}\right).
\end{align*}

According to Lemma 3.5 of Ref.~\cite{chen2021concentration}, for Hamiltonian $\oH=\sum_{l=1}^L\oH_l$ and $\overline{U}_i$, the deterministic bias term is bounded by 
\begin{align*}
\norm{{(\E[\overline{U}_i])^r}-\overline{V}}\leq r\norm{\E[\overline{U}_i]-\overline{V}^{1/r}}\leq O\left(\left(\frac{L}{r}\right)^{p}\cdot(t\lambda_{\oh})^{p+1}\right).
\end{align*}

Here, $\overline{U}^{1/r}=e^{-iHt/r}=e^{iH\delta}$.

The next step is to provide a bound for the random fluctuation term $\norm{\overline{U}_r\ldots\overline{U}_1-{(\E[\overline{U}_i])^r}}$. We still use the Freedman inequality for martingales, and define $B_k=(\E[\overline{U}_i])^{r-k}\overline{U}_k\cdots \overline{U}_1$, $B_0=(\E[\overline{U}_i])^{r}$ and $C_k=B_k-B_{k-1}$. For the random permutation approach, each $C_k$ is bounded by
\begin{align*}
\norm{C_k}\leq\norm{\overline{U}_i-\E[\overline{U}_i]}{\leq}\norm{\overline{U}_i-\overline{V}^{1/r}}+\norm{\E[\overline{U}_i]-\overline{V}^{1/r}}=O\left(\left(\frac{tL\lambda_{\oh}}{r}\right)^{p+1}\right).
\end{align*}

Here, the last equation follows from the fact~\cite{childs2019faster,suzuki1991general} that $\norm{\overline{U}_i-\overline{V}^{1/r}}=O\left(\left(\frac{tL\lambda_{\oh}}{r}\right)^{p+1}\right)$ and $\norm{\E[\overline{U}_i]-\overline{V}^{1/r}}=O\left(\left(\frac{t\lambda_{\oh}}{r}\right)^{p+1}L^p\right)$. Therefore, we can derive from \Cref{lem:MartConc} that with probability at least $1-\chi$ the random fluctuation scales as
\begin{align}\label{eq:RanPer-dependence}
\norm{\overline{U}_r\ldots\overline{U}_1-{(\E[\overline{U}_i])^r}}\leq O\left(\sqrt{\frac{(tLh_\lambda')^{2p+2}}{r^{2p+1}}\cdot\left(n+\log(\frac{1}{\chi})\right)}\right).
\end{align}

Finally, we provide the bound for the projection error $\norm{(U_r\ldots U_1-\overline{U}_r\ldots\overline{U}_1)\pdl}$. According to \Cref{lem:ProdDisProj}, we decompose the projection error as follows:
\begin{align*}
\norm{(U_r\ldots U_1-\overline{U}_r\ldots\overline{U}_1)\pdl}\leq r\cdot 5\exp(-\frac{1}{q}(\lambda(\Delta'-\Delta)-\alpha'\delta ML-q\log q)),
\end{align*}
where $\lambda=(2Jdk)^{-1}$, $\alpha=eJ$, and $q(p)$ is the length of the product formula in each step depending only on $p$.

Now, we are ready to derive \Cref{thm:RandPermLow}. Similar to the previous section, we still consider two cases:

In the first case, we consider the average-case performance of the random permutation approach in the low-energy subspace and prove~\eqref{eq:RandPermLowExp}. In this case, we ignore the random fluctuation and compute the step complexity $r$ such that the projection error is of the same scale as the deterministic bias, i.e.,
\begin{align*}
\norm{(U_r\ldots U_1-\overline{U}_r\ldots\overline{U}_1)\pdl}\simeq \norm{{(\E[\overline{U}_i])^r}-\overline{V}}.
\end{align*}

We can derive the bound on $\Delta'$ and $r$ as
\begin{align*}
\Delta'&=\Delta+\frac{1}{\lambda}\alpha \delta ML+\frac{q}{\lambda}\log q+\frac{q}{\lambda}\log\left(\frac{r^{p+1}}{t^{p+1}\lambda_{\oh}^{p+1}L^{p}}\right),\\
r&=O\left(\frac{Lt^{1+1/p}(\Delta+dkJq\log q)^{1+1/p}}{\epsilon^{1/p}}\right)+O\left(\frac{Lt^{1+1/(2p+1)}(dkMJ^2)^{\frac12+\frac{1}{4p+2}}}{\epsilon^{1/(2p+1)}}\right).
\end{align*}

This finishes the proof for~\eqref{eq:RandPermLowExp}, which provides an $O(L^{1/p})$ or $O(L^{1/(2p+1)})$ reduction on the step complexity compared to standard product formulas~\cite{csahinouglu2021hamiltonian}. This means that as $p$ increases, the random permutation will eventually lose its reduction---a similar observation to that in Ref.~\cite{childs2019faster}. The gate complexity requires an $O(L)$ overhead from the Trotter number $r$.

In the second case, we want to further control the random fluctuation. Thus we require
\begin{align*}
\norm{(U_r\ldots U_1-\overline{U}_r\ldots\overline{U}_1)\pdl}\simeq \norm{\overline{U}_r\ldots\overline{U}_1-{(\E[\overline{U}_i])^r}}.
\end{align*}

We can derive the bound on $\Delta'$ and $r$ to ensure that the simulation error is bounded by $\epsilon$ with probability at least $1-\chi$ as
\begin{align}
\Delta'&=\Delta+\frac{1}{\lambda}\alpha \delta ML+\frac{q}{\lambda}\log q+\frac{q}{2\lambda}\log\left(\frac{r^{2p+3}(n+\log(1/\chi))}{t^{2p+2}\lambda_{\oh}^{2p+2}L^{2p+2}}\right),\\
r&=O\left(\frac{(Lt)^{1+1/(2p+1)}(\Delta+dkJq\log q)^{1+1/(2p+1)}(n+\log(1/\chi)^{1/(2p+1)})}{\epsilon^{2/(2p+1)}}\right)\nonumber\\
\label{eq:RandPermNum}&\quad+O\left(\frac{(Lt)^{1+1/(4p+3)}(dkMJ^2)^{\frac{2p+2}{4p+3}}(n+\log(1/\chi)^{1/(4p+3)})}{\epsilon^{2/(4p+3)}}\right),
\end{align}
which finishes the proof for \eqref{eq:RandPermLowProb}. The overhead for controlling randomized fluctuation will quickly disappear as the order $p$ increases. Yet, the advantage of this randomized approach compared to deterministic product formulas would also disappear correspondingly. Similarly, The gate complexity requires an $O(L)$ overhead from the Trotter number $r$.

\paragraph{Doubling the order of product formula. }In Ref.~\cite{cho2022doubling}, Cho, Berry, and Huang proposed an approach to double the order of digital quantum simulations via the randomized product formulas. We also consider the performance of this approach in the low-energy subspace and obtain \Cref{coro:DoubLow}.  We refer to Ref.~\cite{cho2022doubling} for a detailed description of this algorithm. As the algorithm requires the implementation of some random well-designed unitaries $\left\{U_h^{(l)}=\exp(\alpha_{h,l}H_{h}^{(l)})\right\}$ according to some well-designed probability $\{p_{h,l}\}$, we have to assume that $U_h^{(l)}$'s also satisfy the properties of $\exp(-iH_l\delta)$'s. We still consider the error decomposition as follows:
\begin{align*}
&\quad\ \norm{(U_r\ldots U_1-V)\pdl}\nonumber\\
&\leq\norm{(U_r\ldots U_1-\overline{U}_r\ldots\overline{U}_1)\pdl}+\norm{\overline{U}_r\ldots\overline{U}_1-{(\E[\overline{U}_i])^r}}+\norm{{(\E[\overline{U}_i])^r}-\overline{V}}.
\end{align*}

According to Theorem 2 of Ref.~\cite{cho2022doubling} and \textit{the assumption that $r\geq (5^{p/2-1}+\frac12)L\lambda_{\oh}t$}, we can first compute the deterministic bias as
\begin{align*}
\norm{{(\E[\overline{U}_i])^r}-\overline{V}}\leq O\left(\frac{(L\lambda_{\oh}t)^{2p+2}}{r^{2p+1}}\right),
\end{align*}

For the random fluctuation, we still define $B_k=\overline{U}_r\ldots\overline{U}_1$, $B_0=(\E[\overline{U}_i])^r$, and $C_k=B_k-B_{k-1}$. We can compute the norm of $C_k$ as
\begin{align*}
\norm{C_k}\leq\norm{\overline{U}-\E[\overline{U}_i]}\leq O\left(\frac{(L\lambda_{\oh}t)^{p+1}}{r^{p+1}}\right).
\end{align*}
Therefore, we can derive that with probability at least $1-\chi$, the random fluctuation scales as
\begin{align*}
\norm{\overline{U}_r\ldots\overline{U}_1-{(\E[\overline{U}_i])^r}}\leq O\left(\sqrt{\frac{(tL\lambda_{\oh})^{2p+2}}{r^{2p+1}}\cdot\left(n+\log(\frac1\chi)\right)}\right).
\end{align*}

For the projection error, we have
\begin{align*}
\norm{(U_r\ldots U_1-\overline{U}_r\ldots\overline{U}_1)\pdl}\leq r\cdot 5\exp(-\frac{1}{q}(\lambda(\Delta'-\Delta)-\alpha'\delta ML-q\log q)),
\end{align*}
where $\lambda=(2Jdk)^{-1}$, $\alpha=eJ$, and $q(p)$ is the length in each step depending only on $p$.

To prove \Cref{coro:DoubLow}, we consider the step complexity required to ensure the average-case performance and the convergence of the algorithm, respectively. In the first setting, we ignore the random fluctuation and compute the step complexity $r$ such that the projection error is of the same scale as the deterministic bias, i.e.,
\begin{align*}
\norm{(U_r\ldots U_1-\overline{U}_r\ldots\overline{U}_1)\pdl}\simeq \norm{{(\E[\overline{U}_i])^r}-\overline{V}}.
\end{align*}

We can derive the bound on $\Delta'$ and $r$ as
\begin{align*}
\Delta'&=\Delta+\frac{1}{\lambda}\alpha'\delta ML+\frac{q}{\lambda}\log q+\frac{q}{\lambda}\log\left(\frac{r^{2p+2}}{t^{2p+2}\lambda_{\oh}^{2p+2}L^{2p+2}}\right),\\
r&=O\left(\frac{(Lt)^{1+1/(2p+1)}(\Delta+dkJq\log q)^{1+1/(2p+1)}}{\epsilon^{1/(2p+1)}}\right)+O\left(\frac{(Lt)^{1+1/(4p+3)}(dkMJ^2)^{\frac{2p+2}{4p+3}}}{\epsilon^{1/(4p+3)}}\right).
\end{align*}

This finishes the proof for $r_{\text{exp}}$ in \eqref{eq:DoubLowExp}. For $r_{\text{prob}}$, we follow a similar proof to \Cref{thm:RandPermLow} to obtain the same result. The gate complexity required correspondingly is
\begin{align*}
G&=O\left(\frac{L(Lt)^{1+1/(2p+1)}(\Delta+dkJq\log q)^{1+1/(2p+1)}}{\epsilon^{1/(2p+1)}}\right)+O\left(\frac{L(Lt)^{1+1/(4p+3)}(dkMJ^2)^{\frac{2p+2}{4p+3}}}{\epsilon^{1/(4p+3)}}\right).
\end{align*}

\begin{figure}[htbp]
    \centering
    \hspace{.05\linewidth}
    \begin{subfigure}[ht]{.4\linewidth}
        \centering
        \includegraphics[height = 3.3cm]{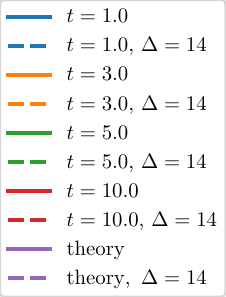}
    \end{subfigure}\hfill
    \begin{subfigure}[ht]{.4\linewidth}
        \centering
        \hspace{1em}\includegraphics[height = 3.6cm]{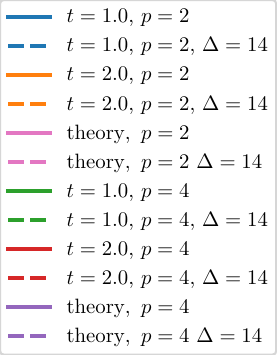}
    \end{subfigure}\\[1.3em]
    \begin{subfigure}[ht]{.47\linewidth}
        \includegraphics[width=\linewidth]{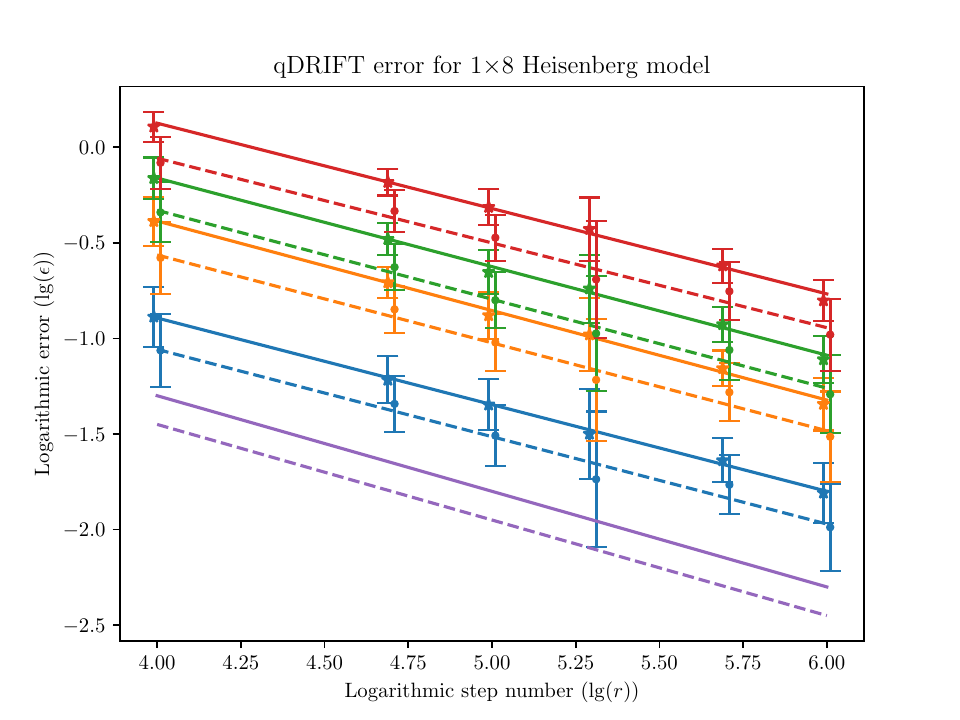}
        \caption{qDRIFT for system size $n=8$}
        \label{fig:qdrift-1x8}
    \end{subfigure}\hfill
    \begin{subfigure}[ht]{.47\linewidth}
        \includegraphics[width=\linewidth]{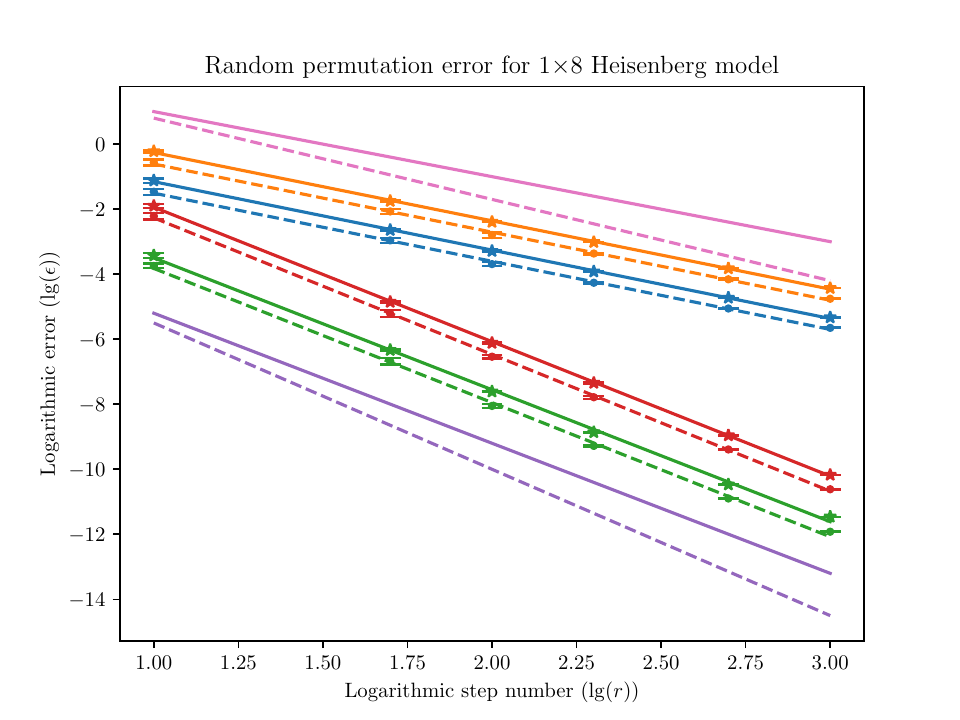}
        \caption{Random permutation for system size $n=8$}
        \label{fig:permutation-1x8}
    \end{subfigure}\\[1em]
    \begin{subfigure}[ht]{.47\linewidth}
        \includegraphics[width=\linewidth]{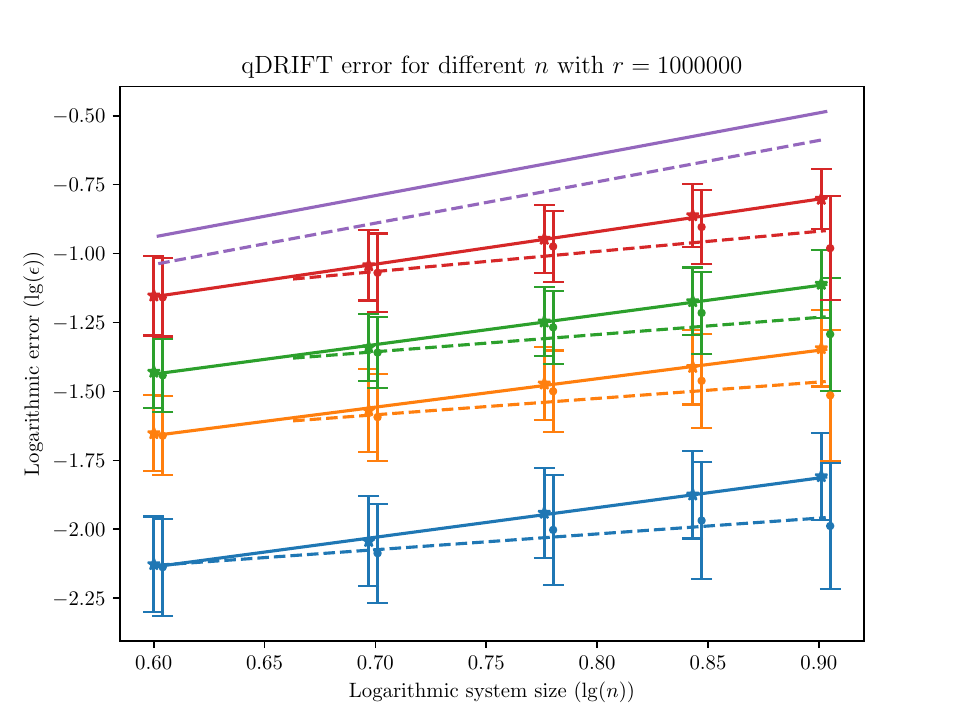}
        \caption{qDRIFT for step number $r=10^6$}
        \label{fig:qdrift-n}
    \end{subfigure}\hfill
    \begin{subfigure}[ht]{.47\linewidth}
        \includegraphics[width=\linewidth]{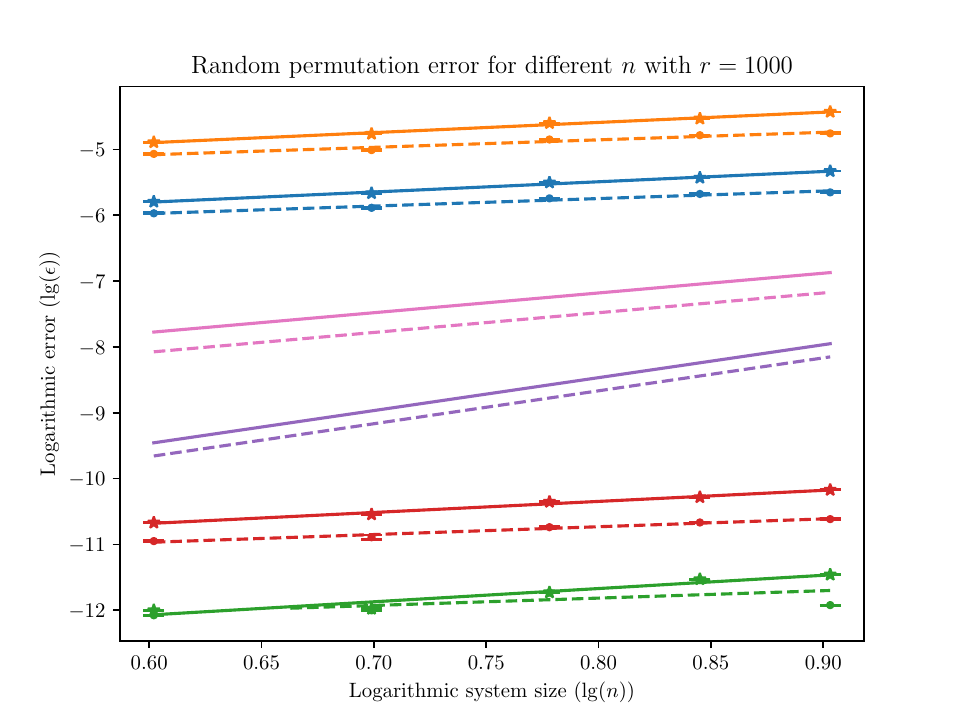}
        \caption{Random permutation for step number $r=10^3$}
        \label{fig:permutation-n}
    \end{subfigure}
    \caption{\textbf{Full Hilbert space vs. low-energy subspace: simulation error induced by the qDRIFT and random permutation algorithms for localized Heisenberg spin-$\mathbf{1/2}$ models.} The solid lines denote errors in the full Hilbert space while the dashed lines denote errors in the low-energy subspace. We distinguish the total evolution time $t$ or the order of the product formula $p$ by the color of the lines. {We also compare the numerical and the theoretical error scalings}. (a) We plot the worst-case and the low-energy simulation errors of the qDRIFT algorithm as a function of step number $r$. (b) The same as the upper left subfigure, but for the random permutation approach at $p=2$ and $p=4$. (c) We plot the worst-case and the low-energy simulation errors of the qDRIFT algorithm as a function of system size $n$. (d) The same as the lower left subfigure, but for the random permutation approach at $p=2$ and $p=4$.}
    \label{fig:Random}
\end{figure}

To benchmark the performance of random Hamiltonian simulation approaches in the low-energy subspace, we further carry out extensive numerical experiments concerning simulating the one-dimensional localized homogeneous Heisenberg models of different system sizes $n\in[4,8]$ without external fields as follows:
\begin{equation}\label{eq:1xn}
    H = \underbrace{-\sum_{i=1}^{n-1}X_iX_{i+1}}_{H_1}\underbrace{-\sum_{i=1}^{n-1}Y_iY_{i+1}}_{H_2}\underbrace{-\sum_{i=1}^{n-1}Z_iZ_{i+1}}_{H_3},
\end{equation}
where $X_i$, $Y_i$, and $Z_i$ are Pauli operators acting on the $i$-th spin, and the summation is over every adjacent spin pair. The Hamiltonian in \eqref{eq:1xn} is decomposed into three interaction terms: $H=H_1+H_2+H_3$, containing all $X$, $Y$, and $Z$ terms, respectively, as indicated by the curly brackets. Since the ground energy $E_{0l}<0$ for every $H_l$, we implement the shift $H_l\rightarrow H_l-E_{0l}I$ for every $H_l$. We empirically pick $\Delta = 14$ (the shifted ground energy of the largest system $n=8$) such that for system sizes $n\in[4,8]$, the low-energy subspace covers neither (nearly) the whole spectrum nor a vacant subspace. {We also provide the theoretical scaling for simulation error in the figures. We remark that as we can not theoretically calculate the constant coefficient, we \emph{cannot} compare the \emph{exact values} of theoretical and numerical simulation error. In the remaining figures, we compare the \emph{scaling} between these two cases instead.}

In \Cref{fig:Random}, we plot the qDRIFT error and the random permutation error as a function of step number $r$ and system size $n$ in four subfigures.
We fix the total evolution time $t$ and randomly sample $10$ sequences for each choice of parameters. For the qDRIFT approach, in order to distinguish between closely spaced error bars, we made a tiny shift to the abscissa of nearby data points. 
We can observe that the projection to the low-energy subspace reduces the error compared to the full Hilbert space for the same parameter configuration. We plot the numerical results for the qDRIFT approach in the left two figures of \Cref{fig:Random}. We estimate that $\epsilon=O(nr^{-1/2})$ for full Hilbert space simulations and $\epsilon=O(n^{2/3}r^{-1/2})$ for low-energy simulations. The numerical performances for both cases have better scaling concerning the $n$ dependence than theoretical bounds as the latter one gives $\epsilon=O(n^{3/2}r^{-1/2})$ and $\epsilon=O(\max\{n^{1/2}r^{-1/2}, n^{3/2}r^{-3/2}\})$ for the two cases as shown in Ref.~\cite{chen2021concentration} and \eqref{eq:qDRIFTNum}. We remark that we take the random fluctuation into consideration here as the variance is much larger than the average value for simulation errors. In the right two figures, we plot the numerical results for the random permutation approach. We estimate that $\epsilon=O(n^{p/2+1}r^{-p})$ and $\epsilon=O(nr^{-p})$ for worst-case and low-energy simulations, which are again better than the theoretical bounds of $\epsilon=O(n^{p+1}r^{-p})$ and $\epsilon=O(\max\{n^{1/2}r^{-p-1/2},n^{p+3/2}r^{-2p-3/2}\})$ in Ref.~\cite{childs2019faster} and \eqref{eq:RandPermNum}. We can conclude that numerical error for the low-energy subspace has a better performance than the theoretical bounds in numerics and provides a steady reduction in the simulation error with various parameters and settings.

\subsection{Symmetry protections}\label{sec:Sym}
In this subsection, we consider an alternative approach, symmetry protection~\cite{tran2021faster}. Given a Hamiltonian $H$, we consider the group of unitary transformations denoted by $\cC$ such that $[H,C]\equiv 0$ for each unitary $C$ chosen from $\cC$. As mentioned in \Cref{sec:Results}, we implement the Lie-Trotter formula $S_1(\delta)$ in each time step with a symmetry transformation $C\in\cC$. Explicitly, the simulation for $V=e^{-iHt}$ reads 
\begin{align*}
U(t)=\prod_{\mu=1}^rC_\mu^\dagger S_1(\delta)C_\mu\approx V(t),
\end{align*}
where $V(t)=e^{-iHt}$ is the ideal evolution. 

We briefly recap the results in Ref.~\cite{tran2021faster} for the intuition of reducing simulation error using symmetry protection. Intuitively, $U(t)$ can be represented as an evolution on a slightly different Hamiltonian $H_{\text{eff}}$ and
\begin{align*}
U(t)=\prod_{\mu=1}^re^{-iC_\mu^\dagger H_{\text{eff}}C_\mu \delta}=\prod_{\mu=1}^re^{-i(H+C_\mu^\dagger \kappa C_k)\delta}\approx e^{-i\left(H+\frac{1}{r}\sum_{\mu=1}^rC_\mu^\dagger \kappa C_\mu\right)t}\approx V(t),
\end{align*}
where $\kappa=H_{\text{eff}}-H$.

For $k$-local Hamiltonians, we define the following quantities~\cite{childs2021theory}:
\begin{align*}
&\alpha=\sum_{\gamma_1=1}^L\sum_{\gamma_2=\gamma_1+1}^L\norm{[H_{\gamma_2},H_{\gamma_1}]}=O(n),\\
&\beta=\sum_{\gamma_1=1}^L\sum_{\gamma_2=\gamma_1+1}^L\sum_{\gamma_3=\gamma_2}^L\norm{[H_{\gamma_3},[H_{\gamma_2},H_{\gamma_1}]]}=O(n),\\
&\gamma=\sum_{\gamma_1=1}^L\sum_{\gamma_2=\gamma_1+1}^L\norm{[H,[H_{\gamma_2},H_{\gamma_1}]]}=O(n^2).
\end{align*}
Ref.~\cite{tran2021faster} obtains the following result concerning the simulation error in each time step:
\begin{lemma}[Lemma 7 of~\cite{tran2021faster}]
For all $\delta$ such that $\beta\delta\leq2\alpha$ and $\alpha^2\delta\leq\beta+\gamma$, the Trotter error in each step is represented by
\begin{align*}
V(\delta)-S_1(\delta)=V(\delta)v_0\frac{\delta^2}{2}+\cV(\delta),
\end{align*}
where $v_0=\sum_{\gamma_1=1}^L\sum_{\gamma_2=\gamma_1+1}^L[H_{\gamma_2},H_{\gamma_1}]$ and $\cV(\delta)\leq\iota \delta^3$ with $\iota=5(\gamma+\beta)/6$.
\end{lemma}
Based on the above result, the Trotterization error for the full Hilbert space simulation is:
\begin{lemma}[Theorem 1 of~\cite{tran2021faster}]\label{lem:SymFull}
The simulation error for the whole evolution using symmetry protection for $k$-local Hamiltonians is bounded by:
\begin{align*}
\norm{V(t)-\prod_{\mu=1}^LC_{\mu}^\dagger S_1(\delta)C_\mu}\leq\overline{v}_0\frac{t^2}{2r}+\iota\frac{t^3}{r^2}+2\left(\frac12\norm{v_0}+\iota\frac{t}{r}\right)^2\frac{t^4}{r^2},
\end{align*}
where 
\begin{align*}
\overline{v}_0=\norm{\frac{1}{r}\sum_{\mu=1}^rC_\mu^\dagger V(\mu s)^\dagger v_0 V(\mu\delta)C_\mu}
\end{align*}
is the averaged first-order error term.
\end{lemma}
We then focus on $\overline{v}_0$. According to Ref.~\cite{tran2021faster}, in general $\norm{\overline{v}_0}\propto\norm{v_0}/r^\theta$ for $\theta\in[0,1]$. Hence, it suffices to choose a step complexity of
\begin{align*}
r=\max\left\{{O}\left(\left(\frac{nt^2}{\epsilon}\right)^{\frac{1}{1+\theta}}\right),{\Tilde{O}}\left(\frac{nt^{3/2}}{\sqrt{\epsilon}}\right)\right\}.
\end{align*}
to bound the simulation error below $\epsilon$. There are several examples given for the value this $\theta$. For example, when we draw symmetry transformations randomly, the behavior of the error would be analogous to a random walk, which results in $\theta=0.5$. The rigorous proof for this intuition for localized Heisenberg model was provided in~\cite{tran2021faster}. There are also some instances when $\theta=1$~\cite{tran2021faster,gong2023improved}.

Now, we consider the performance of this symmetry protection approach in the low-energy subspace, which is also an open question in~\cite{tran2021faster}. Our goal is to compute the following error term
\begin{align*}
\norm{\left(V(t)-\prod_{\mu=1}^rC_\mu^\dagger S_1(\delta)C_\mu\right)\pdl}.
\end{align*}

We start with the following observation:
\begin{proposition}
Suppose $\cS$ is the symmetry group for $H$, then $\cS$ is a subgroup of the symmetry group for $\oH$, {where $\oH = \pdpl H \pdpl$}.
\end{proposition}

We then decompose the error into three parts:
\begin{align*}
 \norm{\left(V(t)-\prod_{\mu=1}^rC_\mu^\dagger S_1(\delta)C_\mu\right)\pdl}&=\norm{(V(t)-\overline{V}(t))\pdl}+\norm{\left(\overline{V}(t)-\prod_{\mu=1}^rC_\mu^\dagger\overline{S}_1(\delta)C_\mu\right)\pdl}\nonumber\\
 &+\norm{\left(\prod_{\mu=1}^rC_\mu^\dagger\overline{S}_1(\delta)C_\mu-\prod_{\mu=1}^rC_\mu^\dagger S_1(\delta)C_\mu\right)\pdl},
\end{align*}
where $\overline{S}_1(\delta)$ is the Lie-Trotter formula for $\oH=\sum_{l=1}^L\oH_l$. The first part is $0$ according to the property of the projector. The second part, according to the results of the symmetry protection approach, gives:
\begin{align*}
\epsilon_2=O\left(\frac{L\Delta' t^2}{r^{1+\theta}}\right)+O\left(\frac{L^2\Delta^{'2}t^3\log{r}}{r^2}\right).
\end{align*}

According to \Cref{lem:ProdDisProj} and the fact that $q=L$ for first-order Trotterizations, we have
\begin{align*}
\norm{(\overline{S}_1(\delta)-S_1(\delta))\pdl}{\leq}5\exp(-\frac{1}{L}(\lambda(\Delta'-\Delta)-\alpha'\delta ML-L\log L)).
\end{align*}

Analogous to the intuition from symmetry transformation, we assume
\begin{align}\label{eq:transformation}
&\quad\ \norm{\left(\prod_{\mu=1}^rS_\mu^\dagger\overline{S}_1(\delta)S_\mu-\prod_{\mu=1}^rS_\mu^\dagger S_1(\delta)S_\mu\right)\pdl}\nonumber\\
&=O\left(r^{1-{\theta}}\exp\left(-\frac{1}{L}(\lambda(\Delta'-\Delta)-\alpha'\delta ML-L\log L)\right)\right)
\end{align}
for some ${\theta}\in[0,1]$. {Here, we remark that the simulation error for the symmetry-protection approach satisfies ${\theta}\in[0,1]$ is a rigorously proved result from Ref.~\cite{tran2021faster} as shown in \eqref{eq:theta}.}

We can derive the step complexity here as $(\lambda=(2Jdk)^{-1})$:
\begin{align*}
r=\tO\left(\left(\frac{L(\Delta+JdkL)t^2}{\epsilon}\right)^{\frac{1}{1+\theta}}+\left(\frac{L^2MJdkt^3}{\epsilon}\right)^{\frac{1}{2+\theta}}+\frac{L(\Delta+JdkL)t^{\frac32}}{\epsilon^{\frac12}}+\frac{t^{\frac54}L(JdkM)^{\frac12}}{\epsilon^{\frac14}}\right),
\end{align*}
which finishes the proof for \Cref{thm:SymProtLow}. The gate complexity of this approach depends on the property of each symmetry transformation $C_\mu$ and varies from case to case. However, we can still observe that the step complexity above is strictly better than that provided by either standard Trotterizations or the symmetry-protected simulation in the full Hilbert space.

\begin{figure}[ht]
    \centering
    \begin{subfigure}[ht]{.18\linewidth}
        \vspace{-1.5em}
        \includegraphics[width=\linewidth]{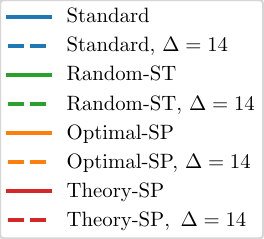}
    \end{subfigure}\hfill
    \begin{subfigure}[ht]{.4\linewidth}
        \includegraphics[width=\linewidth]{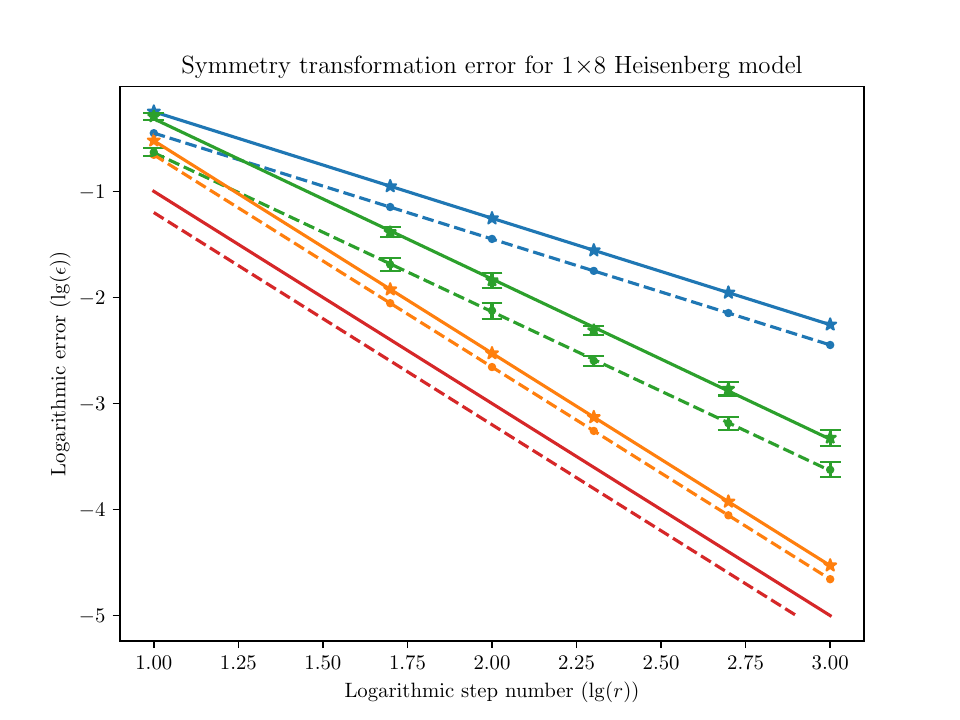}
        \caption{Symmetry transformation for $n=8$}
    \end{subfigure}\hfill
    \begin{subfigure}[ht]{.4\linewidth}
        \includegraphics[width=\linewidth]{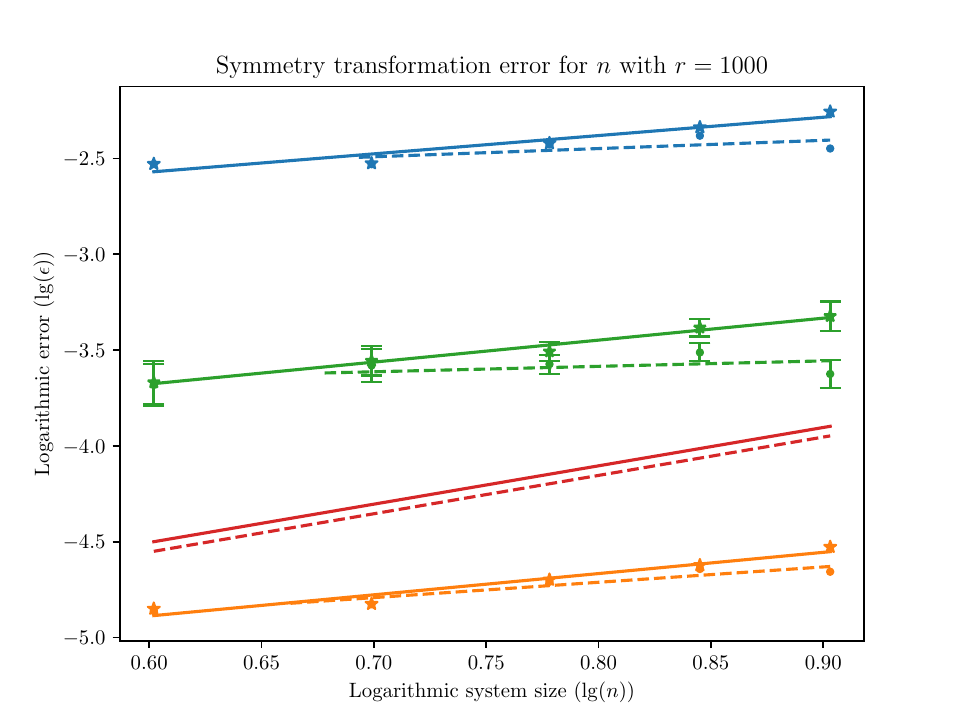}
        \caption{Symmetry transformation for $r=10^3$}
    \end{subfigure}
    \caption{\textbf{Full Hilbert space vs. low-energy subspace: Error induced by the symmetry protection approach for Heisenberg spin-$\mathbf{1/2}$ models.} The solid lines denote errors in the full Hilbert space while the dashed lines denote errors in the low-energy subspace. We distinguish the experiments with different schemes by the color of the lines. {We also compare the numerical and the theoretical error scalings}. (a) We plot the worst-case and the low-energy simulation errors as a function of step number $r$. (b) We plot the worst-case and the low-energy simulation errors as a function of system size $n$.}
    \label{fig:Symmetry}
\end{figure}

We conduct numerical experiments concerning simulating the one-dimensional Heisenberg models of different system sizes $n\in[4,8]$ in \eqref{eq:1xn}. We consider three different schemes of simulations: the standard first-order Trotterization (``Standard"), the first-order Trotterization protected by a random set symmetry transformation (``Random-ST"), and the first-order Trotterization protected by the optimal set (``Optimal-SP"). For the ``Random-ST" approach, considering the $SU(2)$-symmetry of Heisenberg models, each step we sample a $2\times 2$ matrix from norm distribution ensuring diagonal conjugation, then scale the determinant to $1$ and finally tensorize the normalized matrix to the $n$-th power. We randomly repeat this process 10 times and plot the average value as well as error bars. For the ``Optimal-SP" approach, we employ $C_0 = U_H^{\otimes n}, C_\mu = C_0^\mu$ as \cite{tran2021faster}, where $U_H$ denotes single-qubit Hadamard gate.

In \Cref{fig:Symmetry}, we plot the symmetry protection error as a function of step number $r$ and system size $n$ in two subfigures. In our numerical simulations, we fix the total evolution time $t=1.0$. We can observe that the projection to the low-energy subspace reduces the error compared to the full Hilbert space for the same parameter configuration. We estimate that full Hilbert space simulations and low-energy simulations respectively yield $\epsilon= O(nr^{-1})$ and $\epsilon= O(n^{2/5}r^{-1})$ for the standard Trotter formula, $\epsilon= O(nr^{-3/2})$ and $\epsilon= O(n^{1/5}r^{-3/2})$ for the random set symmetry transformation, $\epsilon= O(nr^{-2})$ and $\epsilon= O(n^{4/5}r^{-2})$ for the optimal set symmetry protection. The numerical performances for all the three approaches conform to the scaling of theoretical bounds, which give $\epsilon= O(\max\{nr^{-(1+\theta)},n^2r^{-2}\})$ and $\epsilon= O(\max\{r^{-(1+\theta)},r^{-2}, nr^{-(2+\theta)}, n^2r^{-4}\})$ for the full Hilbert space and low-energy subspace simulations. 

\subsection{Simulating power-law models}\label{sec:Power}

Except for $k$-local Hamiltonians, we also show that simulating low-energy states provides improvement for power-law interactions $H=\sum_{\bi,\bj\in\Sigma}H_{\bi,\bj}$ with exponent satisfying:
\begin{align*}
\norm{H_{\bi,\bj}}\leq\begin{cases}
1 & \quad\text{if }\bi=\bj,\\
\frac{1}{\norm{\bi-\bj}_2^\alpha} & \quad\text{if }\bi\neq\bj,
\end{cases}
\end{align*}
where $\bi,\bj\in\Sigma$ are the qubit sites, $\Lambda\subseteq\R^D$ is a $D$-dimensional square lattice, and $H_{\bi,\bj}$ is an operator supported on two sites $\bi,\bj$. Through direct computation~\cite{childs2021theory}, we have
\begin{align*}
g=\norm{\max_{\bi}\sum_{\bj\neq\bi} H_{\bi,\bj}}\leq\begin{cases}
O(n^{1-\alpha/D})&\quad\text{for }0\leq\alpha<D,\\
O(\log(n))&\quad\text{for }\alpha=D,\\
O(1)&\quad\text{for }\alpha>D.\\
\end{cases}
\end{align*}

Here, $g$ is an upper bound on the strengths of the interactions associated with a single spin. We can observe that every term of the power-law Hamiltonian is $2$-local. As assumed by the \Cref{thm:PowerLow}, each part $H_l$ contains at most {$M$ $k$-local interaction terms}. We consider an operator $A$, and denote $E_A$ as the subset of interaction terms in $H$ that do not commute with $A$ and $R_A$ as the strengths of the terms in $E_A$. Then we can bound $R$ by $R\leq 2g$ and $R\leq 2gn$ for any $H_l$ or $H_l^n$ for arbitrary $n$. We consider the following term:
\begin{align*}
\norm{e^{sH}Ae^{-sH}}.
\end{align*}

Using the Hadamard formula~\cite{miller1973symmetry}, we write
\begin{align*}
\norm{e^{sH}Ae^{-sH}}=A+s[H,A]+\frac{s^2}{2!}[H,[H,A]]+\cdots=\sum_{\ell=0}^\infty\frac{s^\ell}{\ell!}K_\ell.
\end{align*}

For $\ell=0$, we have $\norm{K_\ell}=\norm{A}$, and for $\ell>0$, we have
\begin{align*}
K_\ell\coloneqq\sum_{X_1\in E_A}\sum_{X_2|X_1}\sum_{X_3|X_2,X_1}\cdots\sum_{X_\ell|X_{\ell-1},\cdots,X_1}[h_{X_\ell},[h_{X_{\ell-1}},\cdots,[h_{X_1},A]\cdots]],
\end{align*}
where $h_X$ denotes a term acting on a two-qubit string $X$ and the sum $\sum_{X_j|X_{j-1},\ldots,X_1}$ denotes a summation over the non-zero terms in the commutator. {After shifting each $h_X$ to positive semi-definite operator, we denote $\tilde{h}_X=h_X-\frac{1}{2}\norm{h_X}$. It follows that $\|{\tilde{h}_X}\| \leq \frac{1}{2}\norm{h_X}$ and so $\|{[h_X, O]}\|=\|[\tilde{h}_X, O]\| \leq 2\|{\tilde{h}_X}\| \cdot\|O\| \leq\|{h_X}\|\cdot\|O\|$~\cite{arad2016connecting}.} Thus the norm of $\norm{K_\ell}$ can be bounded by
\begin{align*}
\norm{K_\ell}\leq\sum_{X_1\in E_A}\norm{h_{X_1}}\cdots\sum_{X_\ell|X_{\ell-1}\cdots X_1}\norm{h_{X_\ell}}\norm{A}.
\end{align*}

Now, we compute the value for each term. Notice that the sum of the norms of $h_X$ that do not commute with $A$ is bounded by $R_A$. The sum of the norms of $h_X$ that does not commute with another $h_Y$ is bounded by $2g$. We have
\begin{align*}
\sum_{X_j|X_{j-1},\ldots,X_1}\norm{h_{X_j}}\leq R_A+2(j-1)g.
\end{align*}

Therefore, we have 
\begin{align*}
\norm{K_\ell}\leq \norm{A}R_A(R_A+2g)(R_A+4g)\cdots(R_A+2(\ell-1)g)=(2g)^\ell {m}({m}+1)\cdots({m}+\ell-1)\norm{A},
\end{align*}
where ${m}\coloneqq\frac{R_A}{2g}$. Thus, we have
\begin{align*}
\norm{e^{sH}Ae^{-sH}}\leq\norm{A}\sum_{\ell=0}^\infty\frac{(2gs)^\ell}{\ell!}{m}({m}+1)\cdots({m}+\ell-1)={\norm{A}}\cdot\frac{1}{(1-2gs)^{{m}}}.
\end{align*}

Formally, we can prove the following lemma:
\begin{lemma}
Consider a Hamiltonian with $2$-local terms and suppose the interaction on each qubit is bounded by strength $g$. For any $0\leq s\leq\frac{1}{2g}$, we have $\norm{e^{sH}Ae^{-sH}}\leq\norm{A}(1-2gs)^{-R_A/2g}$, where $g$, $R_A$, and $E_A$ are defined as above.
\end{lemma}
Following the relationship~\cite{arad2016connecting} that
\begin{align*}
\norm{\plpg A\pll}=\norm{\plpg e^{-sH}e^{sH}Ae^{-sH}e^{sH}\pll}\leq\norm{e^{sH}Ae^{-sH}}\cdot e^{-s(\Lambda'-\Lambda)},
\end{align*}
we obtain the following lemma:
\begin{lemma}
Consider a Hamiltonian with $2$-local terms and suppose the interaction on each qubit is bounded by strength $g$. We have
\begin{align*}
\norm{\plpg A\pll}\leq\norm{A}\cdot e^{-\lambda(\Lambda'-\Lambda-2R_A)},
\end{align*}
where $\lambda\coloneqq (4g)^{-1}$.
\end{lemma}

According to the above lemma, we can deduce that
\begin{align*}
\norm{\plpg(H_l)^n\pll}\leq ({MJ})^n e^{-\frac{1}{4g}(\Lambda'-\Lambda-{4gn})}{=}({eMJ})^ne^{-\frac{1}{4g}(\Lambda'-\Lambda)},
\end{align*}
where we assume that $\Lambda'-\Lambda-2R\geq 0$.

Thus, we can deduce that
\begin{align*}
\norm{\plpg e^{-i\delta H_l}\pll}&\leq e^{-\lambda(\Lambda'-\Lambda)}(e^{\alpha M\delta}-1),\\
\norm{\plpg e^{-i\delta\oH_l}\pll}&\leq 3e^{-\lambda(\Lambda'-\Lambda)}(e^{\alpha M\delta}-1).
\end{align*}
for $\alpha={eJ}$ a constant. Compared with \Cref{lem:LeakDiff} and \Cref{lem:LeakExp} in \Cref{app:Lemma}, we find that $k$ is exactly replaced by a constant $2g$. Thus, similar to the proof for \Cref{lem:ProdLow}, we can prove the step complexity in \Cref{thm:PowerLow} as:
\begin{align*}
r=\tO\left(\frac{t^{1+\frac1p}L^{\frac1p}}{\epsilon^{\frac1p}}(\Delta+gq\log q)^{1+\frac 1p}\right)+O\left(\frac{t^{1+\frac{1}{2p+1}}}{\epsilon^{\frac{1}{2p+1}}}(LMg)^{\frac12+\frac{1}{4p+2}}\right).
\end{align*}

It is straightforward to observe that the gate complexity required is
\begin{align*}
G=\tO\left(\frac{t^{1+\frac1p}L^{\frac1p+1}}{\epsilon^{\frac1p}}(\Delta+gq\log q)^{1+\frac 1p}\right)+O\left(\frac{Lt^{1+\frac{1}{2p+1}}}{\epsilon^{\frac{1}{2p+1}}}(LMg)^{\frac12+\frac{1}{4p+2}}\right).
\end{align*}

Compared with the gate count required for full-energy simulation tasks in~\cite{childs2021theory}, our result gets rid of the explicit overhead of $O(n^2)$ gates to implement each Trotter step. When we combine multiple terms in each $H_l$, the overhead can be largely reduced.

Here, we perform numerical experiments to benchmark the product formulas in simulating low-energy states for power-law Hamiltonians. In particular, we study the following 2D power-law Heisenberg spin-$1/2$ model without external fields (as shown in the left-hand side of \Cref{fig:power-law}): 
\begin{equation}\label{eq:power-law}
    H = -\sum_{<i,j>}\frac{1}{\|i-j\|_2^\alpha}(X_iX_j+Y_iY_j+Z_iZ_j),
\end{equation}
where $X_i, Y_i$, and $Z_i$ are Pauli operators acting on the $i$-th spin, and the summation is over every spin pair. 
Since the ground energy $E_{0l}<0$ for every $H_l$, we make a shift $H_l\rightarrow H_l-E_{0l}I$ for every term. We empirically pick $\Delta = 15$ such that the low-energy subspace contains eigenstates suitable.

We plot the Trotter error as a function of single-step evolution time at different decay factors $\alpha$ in the right-hand side of \Cref{fig:power-law}. The Hamiltonian in \eqref{eq:power-law} is decomposed into three interaction terms: the horizontal, the vertical, and the remaining, respectively represented by yellow, green, and red arrows as shown in the left-hand side of \Cref{fig:power-law}. {The target evolution within each time step $V(\delta)$ is approximated by the second-order product formula $S_2(\delta)$.} 
We can observe that the projection to the low-energy subspace significantly reduces the error compared to the full Hilbert space for the same parameter configuration. 
Moreover, in the full Hilbert space, the error increases as $\alpha$ grows, while in the low-energy subspace, the opposite is true. This could be explained by the fact that the norm of the Hamiltonian decreases as $\alpha$ increases. Therefore, $\Delta = 15$ becomes less stringent compared to the whole spectrum of the Hamiltonian as $\alpha$ increases. For a fixed $\Delta$, the proportion of the spectrum in the low-energy subspace increases for larger $\alpha$, which results in a smaller improvement compared to the full Hilbert space simulation. {We further remark that the scaling of the numerical simulation regarding parameters fits perfectly with the theoretical results.}

\begin{figure}[ht]
        \centering
    \begin{minipage}[ht]{.27\linewidth}
        \includegraphics[width=\textwidth]{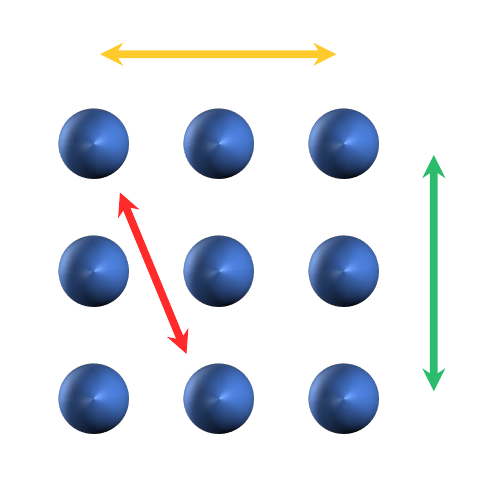}
    \end{minipage}
    \begin{minipage}[ht]{.7\linewidth}
        \includegraphics[width=\textwidth]{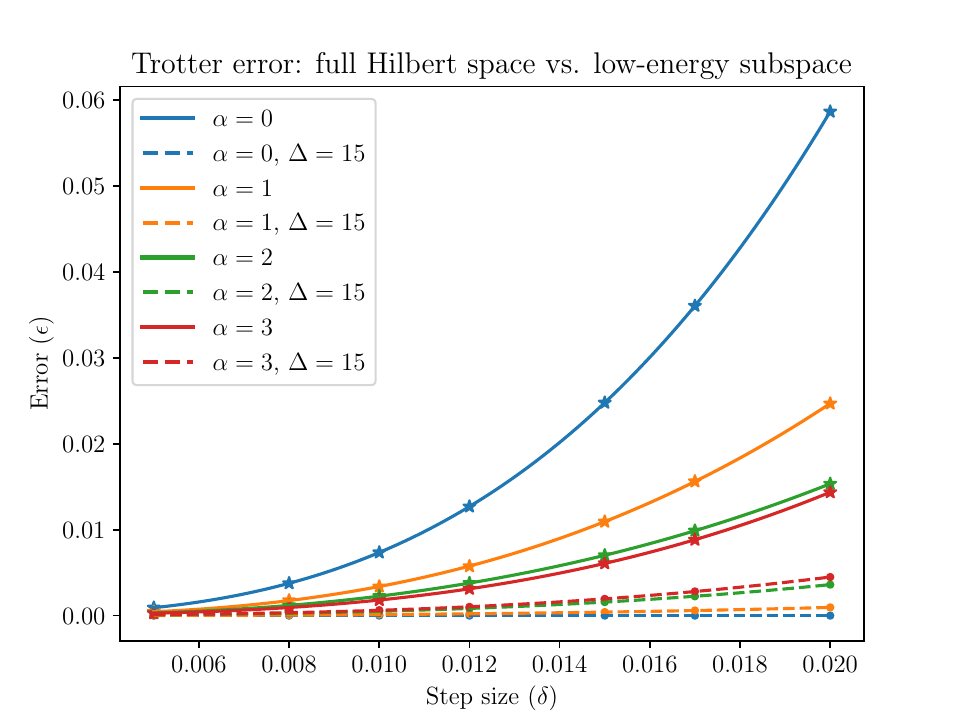}
    \end{minipage}
    \caption{\textbf{Full Hilbert space vs. low-energy subspace: Trotter error induced by the product formulas for a $\mathbf{3\times 3}$ power-law model.} The solid lines denote errors in the full Hilbert space while the dashed lines denote errors in the low-energy subspace, and the plot distinguishes interaction power $\alpha$ by the color of the lines.}
    \label{fig:power-law}
\end{figure}

\section{A Lower Bound for Simulating in the Low-Energy States}\label{sec:lower}
In the above discussion, we show that Hamiltonian simulation in the low-energy subspace can provide improvement compared to full space simulation. It is natural to ask what is the lower bound in the resource requirement for low-energy simulations. Here, we provide a lower bound on the query complexity to simulate any Hamiltonian $H$ for low-energy simulations. In particular, we prove the following theorem:

\begin{theorem}[Lower bound]\label{thm:lower}
Given any positive number $D$, we can construct a $2$-sparse positive semi-definite Hamiltonian $H$ with $\norm{H}=\Theta(D)$ and zero minimal eigenvalues such that simulating some state within the low-energy threshold $\Delta$ and accuracy $\epsilon$ for some chosen scaled time $\tau$ requires at least $\Omega\left(\max\left\{\tau,\frac{\log(\Delta/\epsilon)}{\log\log(\Delta/\epsilon)}\right\}\right)$ queries to $H$.
\end{theorem}

\begin{proof}
For the construction of the Hamiltonian, we adopt the idea of proving the no--fast-forwarding theorem for computing parity using quantum computation~\cite{beals2001quantum,farhi1998limit}. We start with the following Hamiltonian
$H_1$ acting on $(D+1)$-dimensional quantum states. We denote the basis of these quantum states as $\ket{i}$ for $i\in\{0,\ldots,D\}$. The nonzero entries for the Hamiltonian are
\begin{align*}
\bra{i+1}H_1\ket{i}=\bra{i}H_1\ket{i+1}=\frac{\sqrt{(D-i)(i+1)}}{2},\quad i=0,\ldots,D-1.
\end{align*}

We can observe that the Hamiltonian is actually an $J_x$ operator acting on a spin-$D/2$ system. The operator norm of the Hamiltonian is exactly $\frac{D}{2}$. There are two eigenstates corresponding to the eigenvalue $-D/2$ and $D/2$, the former within which is the ground state. In addition, we can find that $\bra{0}e^{iH_1t/2}\ket{D}=(\sin(t/2))^D$.

This Hamiltonian has been utilized as a component to construct lower bounds for Hamiltonian simulation through the following construction $H_2$~\cite{berry2007efficient,berry2014exponential}. In particular, we consider the Hamiltonian $H_2$ generated by an $D$-bit string $x=x_1,\ldots,x_D$ acting on the quantum state $\ket{i,j}$ with $j\in\{0,\ldots,D\}$ and $i\in\{0,1\}$. The nonzero entries of $H_2$ read:
\begin{align*}
\bra{i,j}H_2\ket{i',j+1}=\bra{i,j+1}H_2\ket{i',j+1}=\frac{\sqrt{(D-i)(i+1)}}{2},\quad j=0,\ldots,D-1,
\end{align*}
where $i\oplus i'=x_{i+1}$. This Hamiltonian, however, has a minimal eigenvalue of $-D/2$. In order to make the Hamiltonian positive semi-definite, we consider the following Hamiltonian
\begin{align*}
H=H_2+\frac{D}{2}I,
\end{align*}
which adds an additional $e^{-i\Delta t}$ global phase to the evolution of any input states. As we can always add back the inverse of a global phase to the input state when we fix $t$, we simply ignore the effect of the term $e^{-i\Delta t}$ in the following analysis. By definition, $\ket{0,0}$ is connected to exactly one of $\ket{0,1}$ or $\ket{1,1}$. Moreover, it is connected to the state $\ket{i,j}$ if and only if $x_1\oplus x_2\oplus\ldots\oplus x_i=j$. Therefore, $\ket{0,0}$ is exactly connected to one of $\ket{D,0}$ or $\ket{D,1}$, determining the parity of the bit string $x_1x_2\ldots x_D$. The graph of this Hamiltonian contains two paths within which the vertices are connected, one containing $\ket{0,0}$ and $\ket{D,p(x)}$ for $p(x)$ the parity of $x$ and the other containing $\ket{0,1}$ and $\ket{D,p(x)\oplus 1}$. Given a time $t$, the overlap between $e^{iHt}\ket{0,0}$ and $\ket{D,p(x)}$ is $\abs{\bra{D,p(x)}e^{-iHt}\ket{0,0}}=(\sin(t/2))^D$ while the overlap between $e^{iHt}\ket{0,0}$ and $\ket{D,p(x)\oplus 1}$ is strictly zero. 

Notice that the Hamiltonian is $2$-sparse and we can only query at most two values of $x_i$ by querying $H$ once. Therefore, simulating the evolution of input state $\ket{0,0}$ for time $t=\pi$ yields an unbounded-error algorithm for computing the parity of $x$, thus requiring $\Omega(D)$ queries when we measure the output state in computational basis and consider the probability of obtaining $\ket{D,p(x)}$~\cite{berry2007efficient}. However, there are two issues with this technique in our setting. The first issue is that the state $\ket{0,0}$ is not in the low-energy subspace. It has energy $\bra{0,0}H\ket{0,0}=D/2$ due to the second term $DI/2$. To address this issue, we consider the input state as a mixed state of the minimal eigenstate $\ket{\psi_{\min}}$ and $\ket{0,0}$ as
\begin{align*}
\rho_{\text{in}}=\left(1-\frac{2\Delta}{D}\right)\ket{\psi_{\min}}\bra{\psi_{\min}}+\frac{2\Delta}{D}\ket{0,0}\bra{0,0}.
\end{align*}

This state has energy exactly $\Tr(D\rho_{\text{in}})=\Delta$. The second issue is that, in approximated simulation, we allow an $\epsilon$ error tolerance. Here, we use the techniques in~\cite{berry2014exponential} and consider the choice of $D$ such that the overlap between the evolution of the second term $\frac{2\Delta}{D}\ket{0,0}\bra{0,0}$ and $\ket{D,p(x)}$ the is larger than $\epsilon$, which directly require $\frac{2\Delta}{D}\cdot (\sin(t/2))^D\geq\epsilon$. Even if the simulation error is allowed to be within $\epsilon$, we still have to figure out the parity. In addition, the scaled time $\tau=\norm{H}t=Dt$. Thus we have
\begin{align*}
\left(\sin\left(\frac{\tau}{2D}\right)\right)^D\leq\frac{\Delta\epsilon}{2D}.
\end{align*}

By taking the logarithm at two sides and using the approximation $\sin x\approx x$, we have $D(\log D-\log \tau)=O\left(\frac{1}{\Delta\epsilon}\right)$. Therefore, we have
\begin{align*}
D\geq\Omega\left(\max\left\{\tau, \frac{\log(\Delta/\epsilon)}{\log\log(\Delta/\epsilon)}\right\}\right).
\end{align*}

The lower bound for the number of queries to $H$ is thus $\Omega(D)=\Omega\left(\max\left\{\tau, \frac{\log(\Delta/\epsilon)}{\log\log(\Delta/\epsilon)}\right\}\right)$, which is what we expect.
\end{proof}

\section*{Code Availability}
Github repository: \url{https://github.com/Qubit-Fernand/Digital-Quantum-Simulation}.

\section*{Acknowledgements}
We thank Chi-Fang Chen, Andrew M. Childs, Alexey V. Gorshkov, Minh C. Tran, and Dong Yuan for helpful discussions. SZ and TL were supported by the National Natural Science Foundation of China (Grant Numbers 62372006 and 92365117), and The Fundamental Research Funds for the Central Universities, Peking University.

\bibliographystyle{quantum}
\bibliography{QSimLowEnergy_Final}


\clearpage
\appendix
\section{Auxiliary Lemmas}\label{app:Lemma}
In the following, we carefully pick $\Delta'\geq\Lambda'\geq\Lambda\geq0$ as some value depends on $\Delta$ and $\Delta'\geq\Delta$. We denote $\oH=\pdpl H\pdpl$ and $\oH_l=\pdpl H_l\pdpl$. Therefore, $\norm{\oH_l}\leq\norm{\oH}\leq\Delta'$. 
\begin{lemma}[Lemma 1 of~\cite{csahinouglu2021hamiltonian}]\label{lem:LeakExp}
For $k$-local Hamiltonian $H=\sum_{l=1}^LH_l$ with $H_l\geq 0$ and $\Lambda'\geq\Lambda\geq0$. Then, for all $s\in\R$ and $l=1,\ldots,L$,
\begin{align*}
\norm{\plpg e^{-i\delta H_l}\pll}\leq e^{-\alpha_1(\Lambda'-\Lambda)/J}(e^{\alpha_2 J\abs{\delta}n}-1),
\end{align*}
where $\alpha_1>0$ and $\alpha_2>0$ are constants.
\end{lemma}

\begin{lemma}[Lemma 2 of~\cite{csahinouglu2021hamiltonian}]\label{lem:LeakDiff}
For $k$-local Hamiltonian $H=\sum_{l=1}^LH_l$ with $H_l\geq 0$ and $\Delta'\geq\Lambda'\geq\Lambda\geq0$. Then, for all $\delta\in\R$ and $l=1,\ldots,L$,
\begin{align*}
\norm{\plpl (e^{-i\delta\oH_l}-e^{-i\delta H_l})\pll}&\leq e^{-\alpha_1(\Lambda'-\Lambda)/J}(e^{\alpha_2 J\abs{\delta}n}-1),\\
\norm{\plpg e^{-i\delta\oH_l}\pll}&\leq 3e^{-\alpha_1(\Lambda'-\Lambda)/J}(e^{\alpha_2 J\abs{\delta}n}-1),
\end{align*}
where $\alpha_1>0$ and $\alpha_2>0$ are constants.
\end{lemma}

Based on these two lemmas, we consider approximating errors for product formulas in the low-energy subspaces. We first give the following generic product formulas of $q>1$ terms:
\begin{align*}
W(\bm{\delta})&=e^{-i\delta_q H_{l_q}}\cdots e^{-i\delta_1 H_{l_1}},\\
\oW(\bm{\delta})&=e^{-i\delta_q \oH_{l_q}}\cdots e^{-i\delta_1 \oH_{l_1}},
\end{align*}
where $\bm{\delta}=\delta_1,\ldots,\delta_q,\delta_j\in\R$ and $1\leq l_j\leq L$. Given $\blambda=(\Lambda_1,\ldots,\Lambda_q),\Lambda_j\geq 0$, we define
\begin{align*}
W^\blambda (\bm{\delta})&=\pl{q}e^{-i\delta_q H_{l_q}}\cdots\pl{1}e^{-i\delta_1 H_{l_1}},\\
\oW^\blambda(\bm{\delta})&=\pl{q}e^{-i\delta_q \oH_{l_q}}\cdots\pl{1}e^{-i\delta_1 \oH_{l_1}}.\\
\end{align*}
Regarding these product formulas, the following lemma is known:
\begin{lemma}[Corollaries 1, 2, 3, and 4 of~\cite{csahinouglu2021hamiltonian}]\label{lem:ProdDisProj}
Let $\tau>0$, $\Delta\geq0$, $\lambda=(2Jdk)^{-1}$, and $\alpha=eJ$. Then, if $\blambda$ satisfies $\Lambda_j-\Lambda_{j-1}\geq\frac{1}{\lambda}(\alpha\abs{\delta_j}M+\log(q/\tau))$, $\Lambda_0=\Delta$, and $\Delta'\geq\Lambda_q$, we have
\begin{align*}
&\norm{(W^\blambda(\bm{\delta})-W(\bm{\delta}))\pdl}\leq\tau,\\
&\norm{(\oW^\blambda(\bm{\delta})-W^\blambda(\bm{\delta}))\pdl}\leq\tau,\\
&\norm{(\oW(\bm{\delta})-\oW^\blambda(\bm{\delta}))\pdl}\leq3\tau.
\end{align*}
Combining the three inequalities here, we obtain the result. For $\tau>0$, $\Delta\geq 0$, $\lambda=(2Jdk)^{-1}$, $\alpha=eJ$, and $\abs{\bm{\delta}}=\sum_{j=1}^q\abs{\delta_j}$. Then, if $\Delta'\geq\Delta+\frac{1}{\lambda}(\alpha\abs{\bm{\delta}}M+q\log(q/\tau))$, we have
\begin{align*}
\norm{(\oW(\bm{\delta})-W(\bm{\delta}))\pdl}\leq5\tau.
\end{align*}
\end{lemma}

\begin{lemma}[Corollary 3.4 of~\cite{chen2021concentration}]\label{lem:MartConc}
Given a martingale $\{B_k:k=0,\ldots,r\}\in\mathbb{C}_{{2^n}\times {2^n}}$, we assume that $C_k=B_k-B_{k-1}$ satisfies $\norm{C_k}\leq R_C$ and $\norm{\sum_{k=1}^r\E_{k-1}C_kC_K^\dagger}\leq v$. Here $\E_{k-1}$ is the average conditioned on $B_{k-1},\ldots,B_1$. Then, we have
\begin{align*}
\Pr[\norm{B_r-B_0}\geq \epsilon]\leq {2^{n+1}}\exp\left(\frac{3\epsilon^2}{6v+2R_C\epsilon}\right).
\end{align*}
\end{lemma}

\end{document}